\documentclass[
  aps,                  
  prd,                  
  reprint,              
  nofootinbib,          
  superscriptaddress,   
  showpacs,             
  showkeys              
]{revtex4-1}

\usepackage[utf8]{inputenc}
\usepackage{amsmath,amssymb,amsthm} 
\usepackage{bm,dsfont} 
\usepackage{url}
\usepackage{graphicx}
\usepackage{multirow}
\usepackage{siunitx}

\usepackage{tikz}

\usetikzlibrary{
  babel,
  quotes,
  arrows.meta,
  shapes.symbols,
  shadows.blur,
}

\tikzset{
  > = LaTeX,
  pics/lisa/.style = {
    code = {
      \draw [densely dashed, thick] (0, 1) -- (210:1) -- (330:1) -- cycle;
      \foreach \th in {90, 210, 330} {
        \draw [thick, fill = white] (\th:1) circle [radius = 0.2];
      }
    }
  }
}

\usepackage{hyperref}
\hypersetup{linktocpage,colorlinks=true,urlcolor=blue,linkcolor=blue,citecolor=blue}

\theoremstyle{plain}
\newtheorem{thm}{Theorem}
\newtheorem{lem}[thm]{Lemma}

\theoremstyle{definition}
\newtheorem{defn}[thm]{Definition}

\newcommand{\euler}{\ensuremath{\mathrm{e}}} 
\newcommand{\nwse}[3]{\ensuremath{#1^{#2}_{\phantom{#2} #3}}}
\newcommand{\swne}[3]{\ensuremath{#1_{#2}^{\phantom{#2} #3}}}

\newcommand{\hodge}{{\ast}} 
\newcommand*{\diag}{\operatorname{diag}}
\newcommand{\diff}[1]{\text{d}#1}
\newcommand{\Diff}[1]{\text{D}#1}

\newcommand{\udec}{Departamento de Fí­sica, Universidad de Concepción, Casilla 160-C, 4070105 Concepción, Chile}
\newcommand{\unal}{Departamento de Fí­­sica, Universidad Nacional de Colombia, 111321 Bogotá, Colombia}
\newcommand{\unap}{Facultad de Ciencias, Universidad Arturo Prat, 1110939 Iquique, Chile}
\newcommand{\asctp}{Arnold Sommerfeld Center for Theoretical Physics, Ludwig-Maximilians-Universität München, Theresienstraße 37, 80333 Munich, Germany}
\newcommand{\usach}{Departamento de Física, Universidad de Santiago de Chile, Avenida Ecuador 3493, Estación Central, 9170124 Santiago, Chile}
\newcommand{\ucn}{Departamento de Enseñanza de las Ciencias Básicas, Universidad Católica del Norte, Larrondo 1281, 1781421 Coquimbo, Chile}
\newcommand{\cecs}{Centro de Estudios Científicos (CECs), Avenida Arturo Prat 514, 5110466 Valdivia, Chile}
\newcommand{\imcas}{Institute of Mathematics of the Czech Academy of Sciences, Žitná 25, 11567 Praha 1, Czechia}

\begin{document}

\title{Luminal Propagation of Gravitational Waves in Scalar-tensor Theories:\\The Case for Torsion}

\affiliation{\unap}

\author{José Barrientos}
\email{josebarrientos@udec.cl}
\affiliation{\udec}
\affiliation{\ucn}
\affiliation{\imcas}

\author{Fabrizio Cordonier-Tello}
\email{f.cordonier@physik.uni-muenchen.de}
\affiliation{\asctp}

\author{Cristóbal Corral}
\email{cristobal.corral@usach.cl}
\affiliation{\usach}

\author{Fernando Izaurieta}
\email{fizaurie@udec.cl}
\affiliation{\udec}

\author{Perla Medina}
\email{perlamedina@udec.cl}
\affiliation{\udec}
\affiliation{\cecs}

\author{Eduardo Rodríguez}
\email{eduarodriguezsal@unal.edu.co}
\affiliation{\unal}

\author{Omar Valdivia}
\email{ovaldivi@unap.cl}
\affiliation{\unap}

\date{\today}

\begin{abstract}
Scalar-tensor gravity theories with a nonminimal Gauss--Bonnet coupling typically lead to an anomalous propagation speed for gravitational waves, and have therefore been tightly constrained by multimessenger observations such as GW170817/GRB170817A. In this paper we show that this is not a general feature of scalar-tensor theories, but rather a consequence of assuming that spacetime torsion vanishes identically. At least for the case of a nonminimal Gauss--Bonnet coupling, removing the torsionless condition restores the canonical dispersion relation and therefore the correct propagation speed for gravitational waves. To achieve this result we develop a new approach, based on the first-order formulation of gravity, to deal with perturbations on these Riemann--Cartan geometries.
\end{abstract}

\pacs{04.50.+h}

\keywords{Nonvanishing Torsion, Gravitational Waves, Riemann--Cartan Geometry, Gauss--Bonnet Coupling}

\preprint{LMU-ASC 33/19} 

\maketitle




\section{Introduction}
\label{Sec_Intro}

The multimessenger measurements of the GW170817 event by the LIGO/Virgo Collaboration~\cite{TheLIGOScientific:2017qsa} and the gamma-ray burst GRB~170817A by Fermi and other observatories~\cite{GBM:2017lvd} have provided a strong limit of about one part in \num{e15} on the difference between the propagation speed of gravitational waves (GW) and the speed of light~\cite{Monitor:2017mdv,Huerta2019}.
This observation imposes severe constraints on different viable alternatives to general relativity (GR) aimed at explaining the dark sector of the Universe by means of degrees of freedom beyond the metric ones. In particular, many scalar-tensor theories of the Horndeski/Galileon type predicted, at least in some regimes, an anomalous propagation speed for GWs~\cite{Bettoni:2016mij,Ezquiaga:2017ekz,Ezquiaga:2018btd,Baker:2017hug,Sakstein:2017xjx,Heisenberg:2017qka,Kreisch:2017uet,Nojiri:2017hai}, and even in some cases an anomalous propagation speed for sound waves in Earth's atmosphere~\cite{Mukherjee:2017fqz, Babichev:2018rfj}. This observation implies that, depending on the type of coupling that the scalar fields develop with the geometry, some of these theories have been disfavored by the observational data.

A particular interaction that has been widely studied in the literature is the coupling of scalar fields to topological invariants, e.g., the Pontryagin or Gauss--Bonnet (GB) terms, motivated by effective field theories, string theory, and particle physics~\cite{Alexander:2009tp}. From a phenomenological viewpoint, the scalar-Pontryagin modification to GR---also known as Chern--Simons modified gravity---is an interesting extension that might explain the flat galaxy rotation curves dispensing with dark matter~\cite{Konno:2008np}, while leaving the propagation speed of GWs unaffected~\cite{Nishizawa:2018srh}. This interaction generates nontrivial effects when rotation is included~\cite{Konno:2007ze,Grumiller:2007rv,Konno:2009kg,Yunes:2009hc,Ahmedov:2010fz,Brihaye:2016lsx}, providing a smoking gun in future GW detectors~\cite{Alexander:2007zg,Alexander:2007kv,Yunes:2016jcc,Alexander:2017jmt}. The couplings between scalar fields and the GB term, on the other hand, have been studied in different setups and several solutions that exhibit spontaneous scalarization have been reported~\cite{Gurses:2008zz,Granda:2011eh,Granda:2011ja,Granda:2011kx,Kanti:2015dra,Doneva:2017bvd,Silva:2017uqg,Antoniou:2017hxj,Antoniou:2017acq,Doneva:2017duq,Doneva:2018rou,Bakopoulos:2018nui,Myung:2019wvb,Antoniou:2019awm}. Their stability, however, depends on the choice of the coupling between the scalar field and the GB term~\cite{Myung:2018iyq,Blazquez-Salcedo:2018jnn,Silva:2018qhn}. In spite of this, the theory is experimentally disadvantaged from an astrophysical viewpoint, since it develops an anomalous propagation speed for GWs~\cite{Gong:2017kim}.

Scalar-tensor theories have been formulated in geometries that depart from the pseudo-Riemannian framework several times in the past. In particular, the gravitational role of Riemann--Cartan (RC) geometries, characterized by curvature and torsion, was first discussed by Cartan and Einstein themselves~\cite{EinsteinCartanLetters}, and later on in the framework of gauge theories of gravitation~\cite{Kib61,Sciama:1964wt,Hehl76,Blagojevic:2013xpa}.
Within its simplest formulation---the Einstein--Cartan--Sciama--Kibble (ECSK) theory---torsion is a nonpropagating field sourced only by the spin density of matter.
The nonminimal coupling of scalar fields to geometry dramatically changes this conclusion.
As shown in Ref.~\cite{Valdivia:2017sat}, the typical Horndeski/Galileon couplings and second-order derivatives are generic sources of torsion, even in the absence of any spin density.
When scalar fields are coupled to the Nieh--Yan topological invariant~\cite{Nieh:1981ww}, a regularization procedure of the axial anomaly in RC spacetimes can be prescribed~\cite{Mercuri:2009zi,Chandia:1997hu,Obukhov:1997pz,Kreimer:1999yp,Chandia:1999az}, and a torsion-descendent axion that might solve the strong $CP$ problem in a gravitational fashion is predicted~\cite{Lattanzi:2009mg,Castillo-Felisola:2015ema,Karananas:2018nrj}. The nonminimal coupling to the Gauss--Bonnet invariant, on the other hand, can be motivated from dimensional reduction of string-generated gravity models~\cite{Castillo-Felisola:2016kpe}, and it could drive the late-time acceleration of the Universe in the absence of the cosmological constant~\cite{Toloza:2013wi,Espiro:2014uda,Cid:2017wtf}. The first-order formulation of Chern--Simons modified gravity produces interesting phenomenology when coupled with fermions~\cite{Alexander:2008wi}, and it has been shown that the different nonminimal couplings support four-dimensional black string configurations in vacuum, possessing locally $\text{AdS}_3\times\mathbb{R}$ geometries with nontrivial torsion~\cite{Cisterna:2018jsx}. Remarkably, some of these models can be regarded as a zero-parameter extension of GR~\cite{Alexander:2019ctv}, whose cosmological implications have been recently studied in Ref.~\cite{Alexander:2019wne}. In general, assuming a torsion-free condition reduces the number of independent fields, making the torsionless theory an entirely different dynamical system from the torsionful one.

In this work, we show that it is only the  \emph{torsionless} version of the scalar-tensor theory based on the scalar-GB coupling that predicts an anomalous propagation speed for GWs.
When torsion is taken into account as a rightful component of geometry, GWs generically propagate at the speed of light, and hence those torsional theories survive unfalsified by multimessenger astronomy. Since the dispersion relation for electromagnetic waves (EMW) also remain unmodified by torsion, both EMWs and GWs move along null geodesics, even on a background with nonvanishing torsion.
This does not mean, however, that torsion is wholly undetectable; as shown in Ref.~\cite{Barrientos:2019msu}, torsion affects the propagation of \emph{polarization} for GWs.\footnote{Similar results have been found in teleparallel gravity theories~\cite{Hohmann:2018wxu,Soudi:2018dhv,Bahamonde:2019ipm} and in $f(R)$ theories with a nonminimal coupling to the Nieh--Yan term~\cite{Bombacigno:2018tbo,Bombacigno:2018tih}.}
Thus, at least for this case, the recent observational data only disfavor the torsionless version of the theory, but not its more general torsional relative.

Our article is organized as follows. Section~\ref{Sec_Lagrangian} presents the main line of reasoning, where we introduce the Lagrangian that defines the theory and give general arguments on why the torsionless version of the GB coupling changes the speed of GWs, while the most general dynamical torsion case does not. Sections~\ref{Sec_Math}--\ref{Sec_GW_Speed} prove this statement in detail. In Sec.~\ref{Sec_Math}, we define some mathematical operators that greatly simplify the analysis of GWs on an RC geometry and study their properties and algebra.
In Sec.~\ref{Sec_Gauge}, we use a Lorentz-covariant version of the Lie derivative to generalize the standard Lorenz gauge fixing for the trace-reversed perturbation to this setting.
Section~\ref{Sec_Scales} describes how to separate low- and high-frequency terms.
We follow the approach of Ref.~\cite{Maggiore:1900zz}, with appropriate modifications for the case of RC geometry.
Section~\ref{Sec_GW_Speed} focuses on the leading high-frequency term to prove that torsion restores the canonical dispersion relation, including speed, for the metric mode of GWs.
The eikonal approximation is used to show that the new torsional mode (variously called ``torsionon''~\cite{Izaurieta:2019dix}, ``roton''~\cite{Hehl:1979xk}, or ``gravity $W$ and $Z$ bosons''~\cite{Boos:2016cey}) propagates interacting with the polarization of the standard metric mode, generalizing the results of Ref.~\cite{Barrientos:2019msu}.
Finally, conclusions and further comments are given in Sec.~\ref{Sec_TheEnd}, while many details on the calculations are provided in Appendix~\ref{Apendice}.


\section{Scalar-tensor model with Gauss--Bonnet coupling}
\label{Sec_Lagrangian}

Let $M$ be a four-dimensional spacetime manifold with metric signature $\left(-,+,+,+ \right)$. We shall consider a scalar-tensor theory whose independent dynamical fields are the vierbein one-form $e^{a} = e^{a}{}_\mu \diff{x^\mu}$,\footnote{The vierbein is related to the spacetime metric $g_{\mu\nu}$ through $g_{\mu\nu}=\eta_{ab}e^{a}{}_\mu e^{b}{}_\nu$, with $\eta_{ab} = \diag\left(-,+,+,+\right)$.} the spin connection one-form $\omega^{ab}=\omega^{ab}{}_{\mu} \diff{x^\mu}$, and a complex zero-form scalar field $\phi$, with $\bar{\phi}$ being its complex conjugate. The Lagrangian four-form describing the scalar-tensor theory with Gauss--Bonnet coupling is given by
\begin{align}
  L & =
  \frac{1}{4 \kappa_{4}} \epsilon_{abcd} R^{ab} \wedge e^{c} \wedge e^{d} 
  \nonumber \\ &\quad
  -\frac{1}{4! \kappa_{4}} \left( \Lambda + \kappa_{4} V \right)
  \epsilon_{abcd} e^{a} \wedge e^{b} \wedge e^{c} \wedge e^{d}
   \nonumber \\ &\quad
  -\mathrm{d} \bar{\phi} \wedge \hodge \mathrm{d} \phi
  -\frac{3}{8 \kappa_{4}} \frac{1}{\Lambda + \kappa_{4} V}
  \epsilon_{abcd} R^{ab} \wedge R^{cd},
  \label{Eq_Lagrangian-0}
\end{align}
where $R^{ab} = \mathrm{d} \omega^{ab} + \nwse{\omega}{a}{c} \wedge \omega^{cb}$ is the Lorentz curvature two-form, and $V$ stands for the scalar field's potential, which is assumed to depend only on the magnitude of $\phi$, i.e., $V = V \left( \left\vert \phi \right\vert \right)$. Throughout this article, we work in the context of RC geometry, meaning that the vierbein and spin connection are considered as independent degrees of freedom, and therefore the torsion two-form $T^{a}=\mathrm{D}e^{a}=\mathrm{d}e^{a}+ \nwse{\omega}{a}{b}\wedge e^{b}$ does not need to vanish. The Lagrangian depends only on first-order derivatives of the spin connection and it does not contain derivatives of the vierbein: we do not include any explicit torsional terms~\cite{Mardones:1990qc}.
The coupling constant $\kappa_{4}$ is related to Newton's gravitational constant $G_N$ through $\kappa_4 = 8\pi G_N$, and the cosmological constant is denoted by $\Lambda$.

The Lagrangian~\eqref{Eq_Lagrangian-0} allows for propagating torsion in vacuum and can be regarded as both a particular case of Horndeski's theory~\cite{Horndeski:1974wa} and as a generalization of dynamical Chern--Simons modified gravity~\cite{Jackiw:2003pm,Alexander:2009tp}, both of which set the torsion equal to zero at the outset (but see Ref.~\cite{Valdivia:2017sat} for the torsional version of Horndeski's theory).
The theory defined by Eq.~\eqref{Eq_Lagrangian-0} actually differs from the standard torsional ECSK theory only in the nonminimal coupling $1/\left( \Lambda + \kappa_{4} V \right)$ with the GB density. When $V = \text{const.}$, this last term becomes a topological invariant proportional to the Euler characteristic and does not contribute to the field dynamics in the bulk, although it becomes relevant in the regularization of Noether charges for asymptotically locally anti-de~Sitter spacetimes~\cite{Aros:1999id,Aros:1999kt}. In the general case, namely $V \neq \text{const.}$, this term contributes to the field equations acting as a source of torsion~\cite{Lattanzi:2009mg,Castillo-Felisola:2015ema,Karananas:2018nrj,Castillo-Felisola:2016kpe,Toloza:2013wi,Espiro:2014uda,Cid:2017wtf,Cisterna:2018jsx,Valdivia:2017sat}.
The particular nonminimal coupling with the GB term we use is but one choice; the results regarding the speed of GWs are still valid even if the $1/\left( \Lambda + \kappa_{4} V \right)$ coupling is replaced by an arbitrary function of (the magnitude of) the scalar field, $f \left( \left\vert \phi \right\vert \right)$.
Our choice, however, has several important algebraic and physical properties which lead to a much more transparent treatment.

To start with, this choice for the nonminimal coupling with the GB term allows the Lagrangian to be written in a much more compact way,
\begin{equation}
  L = -\frac{l^{2}}{8\kappa_{4}}
  \epsilon_{abcd} F^{ab} \wedge F^{cd}
  -\mathrm{d} \bar{\phi} \wedge \hodge \mathrm{d} \phi,
\end{equation}
where
\begin{align}
  \Lambda & = \frac{3}{l^{2}}, &  \euler^{2\sigma} & = 1 + \frac{\kappa_{4}}{\Lambda} V, &  F^{ab} & = \euler^{-\sigma} R^{ab} - \frac{1}{l^{2}} \euler^{\sigma} e^{a} \wedge e^{b}.
\end{align}
The independent stationary variations of $L$ with respect to $e^{a}$, $\omega^{ab}$, $\phi$ and $\bar{\phi}$ yield
\begin{align}
  \delta L & =
  \delta e^{a} \wedge \mathcal{E}_{a} +
  \delta \omega^{ab} \wedge \mathcal{E}_{ab} +
  \delta \phi \bar{\mathcal{E}} +
  \delta \bar{\phi} \mathcal{E} 
  \nonumber \\ &\quad +
  \mathrm{d} \left(
    \delta \omega^{ab} \wedge \mathcal{B}_{ab} +
    \delta \phi \bar{\mathcal{B}} +
    \delta \bar{\phi}\mathcal{B}
  \right),
  \label{eq:deltaL}
\end{align}
where
\begin{align}
  \mathcal{E}_{a} & =
  \frac{1}{2\kappa_{4}} \euler^{\sigma} \epsilon_{abcd} e^{b} \wedge F^{cd} 
  \nonumber \\ &\quad +
  \frac{1}{2} \left(
    Z^{b} \bar{Z}_{a} +
    \bar{Z}^{b} Z_{a} -
    \left\vert Z \right\vert^{2} \delta_{a}^{b}
  \right)
  \hodge e_{b},
  \label{Eq_EOM_Vierbein} \\
  \mathcal{E}_{ab} & =
  -\mathrm{D} \left(
    \frac{l^{2}}{4\kappa_{4}} \euler^{-\sigma} \epsilon_{abcd} F^{cd}
  \right),
  \label{Eq_EOM_Spin_Connection} \\
  \mathcal{E} & =
  \mathrm{d} \hodge \mathrm{d} \phi -
  \frac{1}{2} \frac{\phi}{\left\vert \phi \right\vert}
  \frac{\partial}{\partial \left\vert \phi \right\vert} \left(
    \frac{l^{2}}{4\kappa_{4}} \epsilon_{abcd} F^{ab} \wedge F^{cd}
  \right),
  \label{Eq_EOM_Scalar} \\
  \mathcal{B}_{ab} & = -
  \frac{l^{2}}{4\kappa_{4}} \euler^{-\sigma} \epsilon_{abcd} F^{cd},
  \label{Eq_BC_Spin_Connection} \\
  \mathcal{B} & = -\hodge \mathrm{d} \phi,
  \label{Eq_BC_Scalar}
\end{align}
and\footnote{See Definition~\ref{def:In} in Sec.~\ref{Sec_Math} for the mathematical properties of the operator $-\hodge \left( e^{a} \wedge \hodge \right.$.}
\begin{align}
  Z^{a} & = -\hodge \left( e^{a} \wedge \hodge \mathrm{d} \phi \right), &
  \bar{Z}^{a} & = -\hodge \left( e^{a} \wedge \hodge \mathrm{d} \bar{\phi} \right).
\end{align}
The field equations set $\mathcal{E}_a$, $\mathcal{E}_{ab}$, $\mathcal{E}$, and $\bar{\mathcal{E}}$ to zero on $M$, while the boundary conditions are given by the vanishing of $\mathcal{B}_{ab}$, $\mathcal{B}$, and $\bar{\mathcal{B}}$ on $\partial M$.
Furthermore, our $1/\left( \Lambda + \kappa_{4} V \right)$ GB coupling allows the field equations~(\ref{Eq_EOM_Vierbein})--(\ref{Eq_EOM_Scalar}) to be fully compatible with the boundary conditions~(\ref{Eq_BC_Spin_Connection})--(\ref{Eq_BC_Scalar}).
In particular, $\mathcal{E}_{ab}=\mathrm{D}\mathcal{B}_{ab}$, and the system admits the maximally symmetric solution in vacuum
\begin{align}
  \phi & = \phi_{0}, \\
  R^{ab} & = \frac{1}{l^{2}} \euler^{2\sigma_{0}} e^{a} \wedge e^{b}, \\
  T^{a} & = 0,
\end{align}
where $\phi_{0} = \text{const.}$
and $\sigma_{0} = \sigma \left( \phi_{0} \right)$.
This solution describes a spacetime of constant curvature and zero torsion, where the (constant) scalar field plays no role.

When scalar-tensor gravity theories are treated within the first-order formalism, torsion propagates in vacuum sourced by the derivatives of the scalar fields (for further details, see Ref.~\cite{Valdivia:2017sat}). As a matter of fact, the field equation $\mathcal{E}_{ab} = 0$ [cf.~Eq.~(\ref{Eq_EOM_Spin_Connection})] can be rewritten as
\begin{align}
  T^{p} & = -l^{2} \euler^{-2\sigma}
  \frac{\partial \sigma}{\partial \left\vert \phi \right\vert} \bigg\{
    \frac{1}{2} \epsilon_{abcd} e^{p} \wedge e^{d} \hodge \left(
      \mathrm{d} \left\vert \phi \right\vert \wedge R^{ab} \wedge e^{c}
    \right)
   \nonumber \\ &\quad 
    -\hodge \left[
      2 e_{q} \wedge \hodge \left(
        \mathrm{d} \left\vert \phi \right\vert \wedge R^{pq}
      \right)
    \right]
  \bigg\}.
\end{align}
The propagating nature of torsion becomes manifest in this equation, since its right-hand side possesses derivatives of the torsion through $R^{ab}$. This can be seen from the decomposition of the two-form curvature into their Riemannian and torsional pieces\footnote{We use the notation $\mathring{X}$ to denote the ``torsionless version'' of $X$.}
\begin{equation}
  R^{ab} = \mathring{R}^{ab} +
  \mathring{\mathrm{D}} \kappa^{ab} +
  \nwse{\kappa}{a}{c} \wedge \kappa^{cb},
\end{equation}
where $\mathring{R}^{ab}$ is the canonical Riemann curvature two-form and $\kappa^{ab} = \omega^{ab} - \mathring{\omega}^{ab}$ is the contorsion tensor one-form, related to the two-form torsion as $T^{a} = \nwse{\kappa}{a}{b} \wedge e^{b}$. 

To understand why torsion restores the speed of light of GWs for the scalar-GB coupling, let us go back to the Lagrangian~(\ref{Eq_Lagrangian-0}).
Since we are working in the context of RC geometry, the vierbein and the spin connection are independent fields.
This means that the GB coupling term \emph{does not depend on the vierbein}, namely
\begin{equation}
  \delta_{e} \left(
    \frac{1}{\Lambda + \kappa_{4} V} \epsilon_{abcd} R^{ab} \wedge R^{cd}
  \right)
  = 0.\label{Eq_delta_e_GB=0}
\end{equation}
Therefore, the field equation for the vierbein, $\mathcal{E}_{a} = 0$ [cf.~Eq.~(\ref{eq:deltaL})], is insensitive to its presence.
In fact, this equation can be cast into the Einstein--Hilbert form as
\begin{equation}
  \epsilon_{abcd} R^{bc} \wedge e^{d} =
  \frac{\kappa_{4}}{3} \epsilon_{bcde} e^{b} \wedge e^{c} \wedge e^{d}
  \swne{\mathcal{T}}{a}{e},
  \label{Eq_EOM_Vierbein_EH}
\end{equation}
where $\swne{\mathcal{T}}{b}{a}$ is an effective stress-energy tensor given by
\begin{equation}
  \swne{\mathcal{T}}{b}{a} =
    Z^{a} \bar{Z}_{b} + \bar{Z}^{a} Z_{b}
  -  \left(
    \left\vert Z \right\vert^{2} +
    \frac{3}{\kappa_{4} l^{2}} \euler^{2\sigma}
  \right)
  \delta_{b}^{a}.
  \label{Eq_Stress-Energy-Scalar+Cosmological_Const}
\end{equation}
Since GWs arise from perturbations to $\epsilon_{abcd} R^{ab} \wedge e^{c}$, as in the usual torsionless case, the GB coupling cannot possibly contribute to them.

How does the torsionless condition so dramatically alter the behavior of GWs?
To see why, note that naively imposing $T^{a}=0$ in the field equations [cf.~Eqs.~(\ref{Eq_EOM_Vierbein})--(\ref{Eq_BC_Scalar})] does not lead to the standard torsionless case; instead, we get a constant scalar field.
The torsionless condition is a \emph{constraint} on the geometry, and as such it must be imposed through the addition of a Lagrangian multiplier two-form $M_{a}$ to the Lagrangian~(\ref{Eq_Lagrangian-0}),
\begin{equation}
  L \mapsto L_{M} = L - T^{a} \wedge M_{a}.\label{Eq_L-LM}
\end{equation}
It is this modified Lagrangian, $L_{M}$, which reproduces the standard torsionless dynamics.
The field equations derived from $\delta L_{M} = 0$ read
\begin{align}
  \mathcal{E}_{a}^{\left( M \right)} & =
  \mathcal{E}_{a} - \mathrm{D} M_{a} = 0,
  \label{Eq_EOM_Vierbein_T=0} \\
  \mathcal{E}_{ab}^{\left( M \right)} & =
  \mathcal{E}_{ab} - \frac{1}{2} \left(
    M_{a} \wedge e_{b} - M_{b} \wedge e_{a}
  \right)
  = 0,
  \label{Eq_EOM_Spin_Connection_T=0} \\
  \bar{\mathcal{E}}^{\left( M \right)} & =
  \bar{\mathcal{E}} = 0,
  \label{Eq_EOM_Scalar*_T=0} \\
  \mathcal{E}^{\left( M \right)} & =
  \mathcal{E} = 0,
  \label{Eq_EOM_Scalar_T=0} \\
  T^{a} & = 0.
  \label{Eq_EOM_Multiplier_T=0}
\end{align}
Equation~(\ref{Eq_EOM_Spin_Connection_T=0}) can be solved for the Lagrangian multiplier to find
\begin{equation}
  M^{a} = \hodge \left( 2 e_{b} \wedge \hodge \mathcal{E}^{ba} \right) + \frac{1}{2} e^{a} \hodge \left( e_{b} \wedge e_{c} \wedge \hodge \mathcal{E}^{bc}\right).\label{Eq_L_Multiplier}
\end{equation}
Since $\mathcal{E}_{ab}$ includes the Lorentz curvature two-form $R^{ab}$, the term $\mathrm{D} M_{a}$ in Eq.~(\ref{Eq_EOM_Vierbein_T=0}) turns out to be proportional to derivatives of $R^{ab}$.
It is straightforward to see that such terms in $\mathcal{E}_{a}^{\left( M \right)  }$ make a nonzero leading-order contribution in the eikonal limit for perturbations, and therefore modify their dispersion relation and the GW speed.

The lesson to be learned from this analysis is that imposing the torsionless condition \textit{a~priori} is very different from imposing it \textit{a~posteriori}: the torsionless theory, where $T^{a} = 0$ from the outset, has fewer degrees of freedom and it constitutes therefore a different dynamical system from the full torsional theory. Imposing the torsionless condition on the field equations of the torsional theory implies, in the nonminimally coupled case, reducing the scalar field to triviality.
Even if the Lagrangians for both theories may look superficially identical, they are inherently different theories; as shown above, the torsionless condition amounts to a constraint on the dynamics. In the case of the GB coupling, the price of such a constraint is the anomalous speed for GWs.

While certainly plausible, we still have to rigorously show that Eq.~(\ref{Eq_EOM_Vierbein_EH}) leads to the canonical dispersion relation for GWs, including their speed.
To achieve this goal, we must prove that the torsional terms hidden in Eq.~(\ref{Eq_EOM_Vierbein_EH}) do not change the GW dispersion relation and speed.
One also has to deal with the fact that torsion is a propagating field in the nonminimally coupled theory.
The torsional mode interacts with the standard metric mode, and it is not a~priori obvious whether it changes their speed.

In order to prove this point, in the following sections we provide a complete treatment of GWs on a spacetime with torsion.
The necessary mathematical scaffolding is developed in Sec.~\ref{Sec_Math}.
Then, in Sec.~\ref{Sec_GW_Speed} we come back to Eq.~(\ref{Eq_EOM_Vierbein_EH}) to show that torsion restores the canonical dispersion relation and speed for GWs.


\section{Mathematical Intermezzo}
\label{Sec_Math}

In this section, we introduce the mathematical tools that allow us to describe perturbations and waves in the context of RC geometry.

\subsection{A superalgebra of differential operators}
\label{sec:superalgebra}

The differential operators we define appeared originally in Refs.~\cite{Valdivia:2017sat,Izaurieta:2019dix,Barrientos:2019msu}; here we briefly review them for the benefit of the reader who may be unfamiliar with them.
We also show that these operators form a superalgebra and identify its associated super-Jacobi identity, which, beyond the merely aesthetic, eases the study of the eikonal limit of GWs on an RC geometry.

For the sake of generality, in this section we work on a $d$-dimensional manifold $M$ endowed with an RC geometry and a metric tensor with $\eta_{-}$ negative and $d-\eta_{-}$ positive eigenvalues.
In the rest of the paper we restrict ourselves to $d=4$ and a spacetime signature $\eta_{-} = 1$.
The space of differential $p$-forms on $M$ is denoted by $\Omega^{p} \left( M \right)$.


\begin{defn}[Hodge star operator]
  \label{def:hodge}
  The \emph{Hodge star operator}~\cite{Fla89} is a linear map, $\hodge: \Omega^{p} \left( M \right) \to \Omega^{d-p} \left( M \right)$, that takes a differential $p$-form $\alpha \in \Omega^{p} \left( M \right)$,
  \begin{equation}
    \alpha = \frac{1}{p!} \alpha_{\mu_{1} \cdots \mu_{p}}
    \mathrm{d} x^{\mu_{1}} \wedge \cdots \wedge \mathrm{d} x^{\mu_{p}},
  \end{equation}
  and maps it into its \emph{Hodge dual}, $\hodge \alpha \in \Omega^{d-p} \left( M \right)$, defined by
  \begin{equation}
    \hodge \alpha =
    \frac{\sqrt{\left\vert g \right\vert}}{p! \left( d-p \right)!}
    \epsilon_{\mu_{1} \cdots \mu_{d}}
    \alpha^{\mu_{1} \cdots \mu_{p}}
    \mathrm{d} x^{\mu_{p+1}} \wedge \cdots \wedge \mathrm{d} x^{\mu_{d}},
  \end{equation}
  where $g$ is the determinant of the metric tensor and $\epsilon_{\mu_{1} \cdots \mu_{d}}$ is the totally antisymmetric Levi-Civita pseudotensor.
\end{defn}


\begin{defn}
  \label{def:In}
  The operators $\mathrm{I}^{a_{1} \cdots a_{q}}: \Omega^{p} \left( M \right) \to \Omega^{p-q}\left( M \right)$ act on $p$-forms to produce $\left( p-q \right)$-forms according to the rule\footnote{These operators were first defined in Ref.~\cite{Valdivia:2017sat}, where they were denoted as $\Sigma^{a_{1} \cdots a_{q}}$.}
  \begin{equation}
  \mathrm{I}^{a_{1} \cdots a_{q}} =
  \left( -1 \right)^{\left( d-p \right) \left( p-q \right) + \eta_{-}}
  \hodge \left( e^{a_{1}} \wedge \cdots \wedge e^{a_{q}} \wedge \hodge \right.,
  \label{Eq_Def_Iq}  
  \end{equation}
  where $\hodge$ is the Hodge star operator introduced in Definition~\ref{def:hodge}.
  The most important case is $q=1$,
  \begin{equation}
    \mathrm{I}^{a} = \left( -1 \right)^{d \left( p-1 \right) + \eta_{-}}
    \hodge \left( e^{a} \wedge \hodge \right.,
    \label{eq:Ia}
  \end{equation}
  which acts as a coderivative, satisfying the same sign-corrected Leibniz rule as the exterior derivative.
\end{defn}


\begin{defn}
  \label{def:calD}
  We define $\mathcal{D}_{a}: \Omega^{p} \left( M \right) \to \Omega^{p} \left( M \right)$ to be the derivative operator given by~\cite{Valdivia:2017sat}
  \begin{equation}
    \mathcal{D}_{a} = \left\{ \mathrm{I}_{a}, \mathrm{D} \right\} =
    \mathrm{I}_{a} \mathrm{D} + \mathrm{DI}_{a},
    \label{Eq_Def_Da}
  \end{equation}
  where $\mathrm{I}_{a}$ is the coderivative operator introduced in Definition~\ref{def:In}, with $q=1$, and $\mathrm{D}$ stands for the Lorentz-covariant exterior derivative, $\mathrm{D} = \mathrm{d} + \omega$.
\end{defn}


The $\mathcal{D}_{a}$ derivative plays a major role in the study of GWs on RC geometries.
It satisfies Leibniz's rule (without sign correction) and has many useful properties (see, e.g., Lemmas~\ref{lem:Da-Nabla} and~\ref{lem:D-Da} below).


\begin{lem}
  \label{lem:Da-Nabla}
  Let $\nabla_{\mu} = \partial_{\mu} + \Gamma_{\mu}$ be the usual spacetime covariant derivative for the general (not necessarily torsionless) affine connection $\Gamma^{\rho}_{\mu\sigma}$.
  We have
  \begin{equation}
    \mathcal{D}_{a} = \nabla_{a} + \mathrm{I}_{a} T^{b} \wedge \mathrm{I}_{b},
    \label{Eq_Da-Nabla}
  \end{equation}
  where $\mathrm{I}_{a}$ and $\mathcal{D}_{a}$ are the operators introduced in Definitions~\ref{def:In} and~\ref{def:calD}, and $\nabla_{a} = \swne{e}{a}{\mu} \nabla_{\mu}$. Note that Eq.~(\ref{Eq_Da-Nabla}) implies that $\mathcal{D}_{a}$ and $\nabla_{a}$ coincide in the torsionless case, $\mathring{\mathcal{D}}_{a} = \mathring{\nabla}_{a}$.
\end{lem}


\begin{lem}
  \label{lem:D-Da}
  Equation~(\ref{Eq_Def_Da}) can be inverted to yield
  \begin{equation}
    \mathrm{D} = e^{a} \wedge \mathcal{D}_{a} - T^{a} \wedge \mathrm{I}_{a},
    \label{Eq_D-Da}
  \end{equation}
  where $\mathrm{I}_{a}$ and $\mathcal{D}_{a}$ are the operators introduced in Definitions~\ref{def:In} and~\ref{def:calD}.
\end{lem}


\begin{defn}
  \label{def:Dddag}
  We define the generalized covariant co\-derivative $\mathrm{D}^{\ddag}: \Omega^{p} \left( M \right) \to \Omega^{p-1} \left( M \right)$ by~\cite{Valdivia:2017sat}
  \begin{equation}
    \mathrm{D}^{\ddag} = -\mathrm{I}_{a} \mathrm{DI}^{a},
  \end{equation}
  where $\mathrm{I}_{a}$ is the coderivative operator introduced in Definition~\ref{def:In}, with $q=1$, and $\mathrm{D}$ stands for the Lorentz-covariant exterior derivative, $\mathrm{D} = \mathrm{d} + \omega$.
\end{defn}


\begin{defn}[Generalized de~Rham--Laplace wave operator]
  \label{def:GdRL}
  We define the \emph{generalized de~Rham--Laplace wave operator} $\blacksquare_{\text{dR}}: \Omega^{p} \left( M \right) \to \Omega^{p} \left( M \right)$ by~\cite{Valdivia:2017sat}
  \begin{equation}
    \blacksquare_{\mathrm{dR}} = \mathrm{D}^{\ddag} \mathrm{D} + \mathrm{DD}^{\ddag},
  \end{equation}
  where $\mathrm{D}^{\ddag}$ is the generalized covariant coderivative introduced in Definition~\ref{def:Dddag} and $\mathrm{D}$ stands for the Lorentz-covariant exterior derivative, $\mathrm{D} = \mathrm{d} + \omega$.
\end{defn}


\begin{defn}[Generalized Beltrami--Laplace wave operator]
  \label{def:GBL}
  We define the \emph{generalized Beltrami--Laplace wave operator} $\blacksquare_{\mathrm{B}}: \Omega^{p} \left( M \right) \to \Omega^{p} \left( M \right)$ by~\cite{Valdivia:2017sat}
  \begin{equation}
    \blacksquare_{\mathrm{B}} = -\mathcal{D}_{a} \mathcal{D}^{a},
  \end{equation}
  where $\mathcal{D}_{a}$ is the derivative operator introduced in Definition~\ref{def:calD}.
\end{defn}


\begin{lem}
 \label{lem:GenWeitzen}
  The operators introduced in Definitions~\ref{def:GdRL} and~\ref{def:GBL} satisfy the following \emph{generalized Weitzenböck identity for an RC geometry}:
  \begin{equation}
    \blacksquare_{\mathrm{dR}} = \blacksquare_{\mathrm{B}} + \mathrm{I}_{a} \mathrm{D}^{2} \mathrm{I}^{a},
  \end{equation}
  where $\mathrm{I}_{a}$ is the coderivative operator introduced in Definition~\ref{def:In}, with $q=1$, and $\mathrm{D}$ stands for the Lorentz-covariant exterior derivative, $\mathrm{D} = \mathrm{d} + \omega$.
   The proof to this Lemma\footnote{The torsionless (pseudo-) Riemannian case of this Lemma has been known for a long time (see, e.g., Ref.~\cite[Ch.~V, Sec.~B.4]{choquet1977analysis}, Ref.~\cite[Ch.~6.3]{nla.cat-vn2659416} and Ref.~\cite{Bini:2003km}), but the original source of this result for torsionless geometries is unknown to the authors. In fact, we have not been able to find any evidence of it ever appearing in the work of the Austrian mathematician Roland Weitzenböck (1885--1955). If the reader knows the real origin of the Weitzenböck identity, we would be glad to be contacted and to learn about its actual attribution.} for the case of RC geometry was given in Refs.~\cite{Valdivia:2017sat,Barrientos:2019msu}.
\end{lem}


\begin{lem}
\label{lem:Com_D_Hodge}
  The $\mathcal{D}_{a}$ derivative introduced in Definition~\ref{def:calD} satisfies the following useful commutation relation with the Hodge star operator:
  \begin{equation}
    \left[ \mathcal{D}_{a}, \hodge \right] = \mathrm{I}_{a} T^{b} \wedge \mathrm{I}_{b} \hodge.
  \end{equation}
\end{lem}


Most importantly, the operators $\mathcal{D}_{a}$, $\mathrm{I}_{a}$, and $\mathrm{D}$ give rise to a \emph{superalgebra} of differential operators, where the curvature and torsion play the role of structure ``constants.''
This makes sense because $\mathrm{I}_{a}$ and $\mathrm{D}$ are odd (``fermionic'') operators,
\begin{align}
  \mathrm{D} & : \Omega^{p} \left( M \right) \to \Omega^{p+1} \left( M \right), \\
  \mathrm{I}_{a} & : \Omega^{p} \left( M \right) \to \Omega^{p-1} \left( M \right),
\end{align}
while $\mathcal{D}_{a}$ is an even (``bosonic'') operator,
\begin{equation}
  \mathcal{D}_{a} : \Omega^{p} \left( M \right) \to \Omega^{p} \left( M \right).
\end{equation}


\begin{thm}
  \label{thm:superalg}
  The operators $\mathcal{D}_{a}$, $\mathrm{I}_{a}$, and $\mathrm{D}$ (defined above) close on themselves, satisfying the following superalgebra:
  \begin{align}
    \left\{ \mathrm{I}_{a}, \mathrm{D} \right\} & = \mathcal{D}_{a},
    \label{Eq_Super(I,D)} \\
    \left\{ \mathrm{I}_{a}, \mathrm{I}_{b} \right\} & = 0,
    \label{Eq_Super(I,I)} \\
    \left\{ \mathrm{D}, \mathrm{D} \right\} & = 2\mathrm{D}^{2},
    \label{Eq_Super(D,D)} \\
    \left[ \mathrm{I}_{a}, \mathcal{D}_{b} \right] & = - \nwse{T}{c}{ab} \mathrm{I}_{c},
    \label{Eq_Super[Ia,Db]} \\
    \left[ \mathrm{D}, \mathcal{D}_{b} \right] & = \mathrm{D}^{2} \mathrm{I}_{b} - \mathrm{I}_{b}\mathrm{D}^{2},
    \label{Eq_Super[D,Db]} \\
    \left[ \mathcal{D}_{a}, \mathcal{D}_{b} \right] & = \mathrm{I}_{ab} \mathrm{D}^{2} +
    \mathrm{D}^{2} \mathrm{I}_{ab} + \mathrm{I}_{a} \mathrm{D}^{2} \mathrm{I}_{b} -
    \mathrm{I}_{b} \mathrm{D}^{2} \mathrm{I}_{a}
     \nonumber \\ &\quad - \left(
      \mathrm{D} \nwse{T}{c}{ab} \wedge \mathrm{I}_{c} + \nwse{T}{c}{ab} \mathcal{D}_{c}
    \right),
    \label{Eq_Super[Da,Db]}
  \end{align}
  where $\mathrm{D}^{2}$ acts not as a \emph{differential} operator but as one that, by virtue of the Bianchi identities, gives rise to terms proportional to the Lorentz curvature two-form; e.g., $\mathrm{D}^{2} e^{a} = \nwse{R}{a}{b} \wedge e^{b}$, $\mathrm{D}^{2} \nwse{R}{a}{b} = 0$, and $\mathrm{D}^{2} T^{a} = \nwse{R}{a}{b} \wedge T^{b}$.
\end{thm}


\begin{proof}
  The (anti)commutation relations~(\ref{Eq_Super(I,D)})--(\ref{Eq_Super[D,Db]}) are all straightforward to prove. To show that Eq.~(\ref{Eq_Super[Da,Db]}) holds, it suffices to notice that $\left[ \mathcal{D}_{a}, \mathcal{D}_{b} \right] = \left[  \mathcal{D}_{a}, \left\{ \mathrm{D}, \mathrm{I}_{b} \right\} \right]$ and to use the super-Jacobi identity
  \begin{equation}
    \left\{ \mathrm{D}, \left[ \mathrm{I}_{b}, \mathcal{D}_{a} \right] \right\} +
    \left[ \mathcal{D}_{a}, \left\{ \mathrm{D}, \mathrm{I}_{b} \right\} \right] -
    \left\{ \mathrm{I}_{b}, \left[ \mathcal{D}_{a}, \mathrm{D} \right] \right\} = 0.  
  \end{equation}
\end{proof}


In as few words as possible, and at the risk of glossing over some important subtleties, one may say that the study of GWs in the context of RC geometry is very similar to the standard Riemannian case, but using the new derivative $\mathcal{D}_{a}$ instead of the standard torsionless spacetime covariant derivative $\mathring{\nabla}_{\mu}$.


\subsection{Lorentz-covariant Lie derivative}
\label{sec:LCLD}

The usual Lie derivative is not Lorentz covariant.
For instance, while the vielbein transforms as a vector under local Lorentz transformations (LLT), its Lie derivative does not.
Since LLTs are an essential part of our construction [e.g., the Lagrangian~(\ref{Eq_Lagrangian-0}) is invariant under this gauge symmetry], we define a modified Lorentz-covariant version of the Lie derivative that fixes this problem.


\begin{defn}[Cartan's formula]
  \label{def:LieDer}
  When acting on a differential $p$-form, the Lie derivative operator along a vector field $\vec{\xi}$ is given by \emph{Cartan's formula}~\cite{Nakahara:2016},
  \begin{equation}\label{cartanformula}
    \pounds_{\xi} = \mathrm{I}_{\xi} \mathrm{d} + \mathrm{dI}_{\xi},
  \end{equation}
  where $\mathrm{I}_{\xi}$ is the \emph{contraction} operator.%
\footnote{Also called the \emph{interior product} and denoted by $\imath_{\xi}$ or $\xi \lrcorner$.
For our purposes, it proves most convenient to write $\mathrm{I}_{\xi}$ as [cf.~Eq.~(\ref{eq:Ia})] $\mathrm{I}_{\xi} = \left(-1\right)^{d\left(p-1\right)+\eta_{-}} \hodge \left( \xi \wedge \hodge \right.$, where $\xi = \xi_{\mu} \mathrm{d} x^{\mu}$ is the one-form dual to the vector field $\vec{\xi} = \xi^{\mu} \partial_{\mu}$.}
\end{defn}


\begin{defn}[Lorentz-covariant Lie derivative]
  \label{def:LCLD}
  When acting on a differential $p$-form that behaves as a \emph{tensor} under LLTs, the \emph{Lorentz-covariant Lie derivative} operator along a vector field $\vec{\xi}$ is given by the formula~\cite{Hehl:1994ue,Obukhov:2006ge,Obukhov:2007se,Obukhov:2007sh,Corral:2018hxi}
  \begin{equation}
    \mathfrak{L}_{\xi} = \mathrm{I}_{\xi} \mathrm{D} + \mathrm{DI}_{\xi},
  \end{equation}
  where $\mathrm{D}$ is the Lorentz-covariant exterior derivative.
  On the other hand, the Lorentz-covariant Lie derivative of the spin connection one-form $\omega^{ab}$ is defined as
  \begin{equation}
    \mathfrak{L}_{\xi} \omega^{ab} = \mathrm{I}_{\xi} R^{ab},
  \end{equation}
  where $R^{ab}$ is the Lorentz curvature two-form.
\end{defn}


For instance, the Lorentz-covariant Lie derivatives of the vielbein $e^{a}$ and the scalar field $\phi$ respectively read
\begin{align}
  \mathfrak{L}_{\xi} e^{a} & =
  \left( \mathrm{I}_{\xi} \mathrm{D} + \mathrm{DI}_{\xi} \right) e^{a} =
  \mathrm{I}_{\xi} T^{a} + \mathrm{D} \xi^{a}, \\
  \mathfrak{L}_{\xi} \phi & =
  \left( \mathrm{I}_{\xi} \mathrm{D} + \mathrm{DI}_{\xi} \right) \phi =
   \mathrm{I}_{\xi} \mathrm{d} \phi.
\end{align}
It is clear that, when acting on $p$-forms that behave as a scalar under LLTs, the Lorentz-covariant Lie derivative reduces to the standard one given by Cartan's formula~\eqref{cartanformula}. 

One can check directly that
\begin{align}
  \pounds_{\xi} e^{a} & =
  \mathfrak{L}_{\xi} e^{a} + \nwse{\lambda}{a}{b} e^{b}, \\
  \pounds_{\xi} \phi & =
  \mathfrak{L}_{\xi} \phi, \\
  \pounds_{\xi} \omega^{ab} & =
  \mathfrak{L}_{\xi} \omega^{ab} - \mathrm{D} \lambda^{ab},
\end{align}
where $\lambda^{ab} = -\mathrm{I}_{\xi} \omega^{ab}$ plays the role of an infinitesimal local Lorentz parameter. This means that the difference between the usual Lie derivative and its Lorentz-covariant version amounts to an infinitesimal LLT.

The Lorentz-covariant Lie derivative is the suitable operator to define black hole entropy as the Noether charge at the horizon in the first-order formalism, since the standard Lie derivative does not produce the correct transformation law for the vierbein at the bifurcation surface~\cite{Jacobson:2015uqa}.


\section{Perturbations on a Riemann--Cartan geometry and Gauge Fixing}
\label{Sec_Gauge}

\subsection{Lie draggings and Lorentz transformations}
\label{sec:LieLorentz}

In the usual torsionless GW treatment, it proves useful to define the trace-reversed version of the metric perturbation,
\begin{equation}
  \bar{h}_{\mu\nu} = h_{\mu\nu} - \frac{1}{2} h g_{\mu\nu},
\end{equation}
and then to perform a wisely chosen infinitesimal diffeomorphism on the metric,
\begin{equation}
  g_{\mu\nu} \mapsto g_{\mu\nu} +
  \mathring{\nabla}_{\mu} \xi_{\nu} +
  \mathring{\nabla}_{\nu} \xi_{\mu},
\end{equation}
in order to arrive at the Lorenz gauge-fixing condition,
\begin{equation}
  \mathring{\nabla}_{\mu} \bar{h}^{\mu\nu} = 0.
  \label{eq:LGF}
\end{equation}
It is not trivial to generalize this procedure for the case of RC geometry.
A first generalization was put forward in Ref.~\cite{Valdivia:2017sat}, but, while correct, it proved to be far from the best choice: the final result was a cumbersome inhomogeneous GW equation with many torsional couplings.
In this section, we use the $\mathcal{D}_{a}$ derivative (see Definition~\ref{def:calD} in Sec.~\ref{Sec_Math}) to provide a generalized Lorenz gauge fixing in an optimal way.

An \emph{infinitesimal Lie dragging} (LD) on the fields of the theory corresponds to
\begin{align}
  \text{LD}: \left\lbrace
    \begin{aligned}
      \delta_{\text{LD}} \left( \xi \right) e^a
        &= \mathfrak{L}_\xi e^a, \\
      \delta_{\text{LD}} \left( \xi \right) \omega^{ab}
        &= \mathfrak{L}_\xi \omega^{ab}, \\
      \delta_{\text{LD}} \left( \xi \right) \phi
        &= \mathfrak{L}_\xi \phi, 
    \end{aligned}
  \right.
\end{align}
where $\mathfrak{L}_{\xi}$ denotes the \emph{Lorentz-covariant Lie derivative} operator along a vector field $\vec{\xi}$ (see Definition~\ref{def:LCLD} in Sec.~\ref{sec:LCLD}).
Since the Lagrangian~(\ref{Eq_Lagrangian-0}) is Lorentz invariant, we have that
\begin{align}
  \mathfrak{L}_{\xi} L & =
  \mathrm{dI}_{\xi} L
  \nonumber \\ & =
  \mathfrak{L}_{\xi} e^{a} \wedge \mathcal{E}_{a} +
  \mathfrak{L}_{\xi} \omega^{ab} \wedge \mathcal{E}_{ab} +
  \mathfrak{L}_{\xi} \phi \bar{\mathcal{E}} +
  \mathfrak{L}_{\xi} \bar{\phi} \mathcal{E} 
  \nonumber \\ &\quad +
  \mathrm{d} \left(
    \mathfrak{L}_{\xi} \omega^{ab} \wedge \mathcal{B}_{ab} +
    \mathfrak{L}_{\xi} \phi \bar{\mathcal{B}} +
    \mathfrak{L}_{\xi} \bar{\phi} \mathcal{B}
  \right).
\end{align}
From this result we conclude that in a Lorentz-invariant Lagrangian, the only important piece of the LD is the one described by the $\mathfrak{L}_{\xi}$ operator. Additionally, under an infinitesimal LLT, the fields transform according to
\begin{align}
  \text{LLT}: \left\lbrace
    \begin{aligned}
      \delta_{\text{LLT}} \left( \lambda \right) e^a
        &= \lambda^{a}{}_b e^b,\\
      \delta_{\text{LLT}} \left( \lambda \right) \omega^{ab}
        &= - \mathrm{D}\lambda^{ab},\\
      \delta_{\text{LLT}} \left( \lambda \right) \phi
        &=0.
    \end{aligned}
  \right.
\end{align}

The commutator of the infinitesimal LDs and LLTs, once applied to any gravitational field, form the Lie algebra 
\begin{align}
 \left[
   \delta_{\text{LLT}} \left( \lambda_1 \right),
   \delta_{\text{LLT}} \left( \lambda_2 \right)
 \right] &=
   \delta_{\text{LLT}} \left( \lambda_3 \right),
 \\
 \left[
   \delta_{\text{LLT}} \left( \lambda \right),
   \delta_{\text{LD}} \left( \xi \right)
 \right] &=
   \delta_{\text{LD}} \left( \tilde{\xi} \right),
 \\
 \left[
   \delta_{\text{LD}} \left( \xi_1 \right),
   \delta_{\text{LD}} \left( \xi_2 \right)
 \right] &=
   \delta_{\text{LLT}} \left( \bar{\lambda} \right) +
   \delta_{\text{LD}} \left( \bar{\xi} \right),
\end{align}
where we have defined $\lambda_{3}^{ab} = \lambda_{1}^{a}{}_c \lambda_{2}^{cb} - \lambda_{2}^{a}{}_c \lambda_{1}^{cb}$, $\tilde{\xi}^a = \lambda^{a}{}_b\xi^b$, $\bar{\lambda}^{ab} = \mathrm{I}_{\xi_2}\mathrm{I}_{\xi_1} R^{ab}$, and $\bar{\xi}^a = \mathrm{I}_{\xi_1}\mathrm{I}_{\xi_2} T^a$. The commutator between two LDs shows that curvature and torsion appear as ``structure functions'' of the algebra. Remarkably, this algebra closes off shell regardless of the dimensionality of the spacetime, its internal group, the field content of the theory, and even in cases with restricted symmetries~\cite{Corral:2018hxi}. Furthermore, the invariance of the Lagrangian~\eqref{Eq_Lagrangian-0} under arbitrary LDs and LLTs implies the Noether identities
\begin{align}
 \Diff{\mathcal{E}_a} &= \mathrm{I}_a T^b\wedge\mathcal{E}_b + \mathrm{I}_a R^{bc}\wedge\mathcal{E}_{bc} + \mathcal{D}_a\phi\,\mathcal{E} + \mathcal{D}_a\bar{\phi}\,\bar{\mathcal{E}}, \\
 \Diff{\mathcal{E}_{ab}} &= e_{[a} \wedge \mathcal{E}_{b]},
\end{align}
respectively, which are also referred to as the contracted Bianchi identities. 

The $\mathfrak{L}_{\xi}$ operator is a well-defined Lorentz-covariant version of the Lie derivative operator, but it still includes some residual Lorentz freedom.
To see this, note that it is possible to write the action of $\mathfrak{L}_{\xi}$ on $e^{a}$ in terms of the $\mathcal{D}_{c}$ derivative as
\begin{equation}
  \mathfrak{L}_{\xi} e^{a} =
  \mathfrak{L}_{\xi}^{+} e^{a} +
  \mathfrak{L}_{\xi}^{-} e^{a},
\end{equation}
with
\begin{align*}
  \mathfrak{L}_{\xi}^{+} e_{a} & =
  \frac{1}{2} e^{b} \left[
    \xi^{c} \left(
      \mathrm{I}_{b} \mathcal{D}_{c} e_{a} +
      \mathrm{I}_{a} \mathcal{D}_{c} e_{b}
    \right) +
    \mathcal{D}_{b} \xi_{a} +
    \mathcal{D}_{a} \xi_{b}
  \right]
  \\ & =
  \frac{1}{2} \left(
    \mathring{\mathcal{D}}_{b} \xi_{a} +
    \mathring{\mathcal{D}}_{a} \xi_{b}
  \right)
  e^{b}
  \\ & =
  \frac{1}{2} \left(
    \mathring{\mathrm{D}} \xi_{a} +
    \mathring{\mathcal{D}}_{a} \xi
  \right),
\end{align*}
\begin{align*}
  \mathfrak{L}_{\xi}^{-} e_{a} & =
  \frac{1}{2} e^{b} \left[
    \xi^{c} \left(
      \mathrm{I}_{b} \mathcal{D}_{c} e_{a} -
      \mathrm{I}_{a} \mathcal{D}_{c} e_{b}
    \right) +
    \mathcal{D}_{b} \xi_{a} -
    \mathcal{D}_{a} \xi_{b}
  \right]
  \\ & =
  -\frac{1}{2} \left[
    \mathcal{D}_{a} \xi_{b} -
    \mathcal{D}_{b} \xi_{a} + \left(
      T_{abc} - T_{bac}
    \right)
    \xi^{c}
  \right] e^{b},
\end{align*}
and where $\xi = \xi_{\mu} \mathrm{d} x^{\mu} = \xi_{a} e^{a}$ is the one-form dual to the vector $\vec{\xi} = \xi^{\mu} \partial_{\mu} = \xi^{a} \vec{e}_{a}$.

Defining the antisymmetric parameter
\begin{equation}
  \tilde{\lambda}_{ab} = \frac{1}{2} \left[
    \mathcal{D}_{a} \xi_{b} -
    \mathcal{D}_{b} \xi_{a} + \left(
      T_{abc} - T_{bac}
    \right)
    \xi^{c}
  \right],
\end{equation}
it is clear that $\mathfrak{L}_{\xi} e^{a}$ contains a residual Lorentz transformation
\begin{equation}
  \mathfrak{L}_{\xi} e^{a} =
  \mathfrak{L}_{\xi}^{+} e^{a} -
  \nwse{\tilde{\lambda}}{a}{b} e^{b}.
\end{equation}
A similar residual Lorentz freedom is found when $\mathfrak{L}_{\xi}$ acts on $\omega^{ab}$ and $\phi$.

Therefore, we consider the final set of modified infinitesimal Lie draggings (MLD) given by $\mathfrak{L}_{\xi}$ and a counter-Lorentz transformation,
\begin{align}
  \text{MLD}: \left\lbrace
    \begin{aligned}
      \delta_{\text{MLD}}(\xi) e^a &= \mathfrak{L}_\xi^{+} e^a =
        \mathfrak{L}_\xi e^a + \tilde{\lambda}^{a}{}_b e^b,
      \\
      \delta_{\text{MLD}}(\xi) \omega^{ab} &= \mathfrak{L}^{+}_\xi \omega^{ab} =
        \mathfrak{L}_\xi\omega^{ab} - \mathrm{D} \tilde{\lambda}^{ab},
      \\
      \delta_{\text{MLD}}(\xi)\phi &= \mathfrak{L}^{+}_{\xi} \phi =
        \mathfrak{L}_{\xi} \phi,
    \end{aligned}
  \right.
\end{align}
with
\begin{align}
  \mathfrak{L}_{\xi}^{+} e_{a} & =
  \frac{1}{2} \left(
    \mathring{\mathrm{D}} \xi_{a} +
    \mathring{\mathcal{D}}_{a} \xi
  \right),
  \\
  \mathfrak{L}_{\xi}^{+} \omega^{ab} & =
  \mathrm{I}_{\xi} R^{ab} -
  \frac{1}{2} \mathrm{D} \left[
    \mathcal{D}^{a} \xi^{b} -
    \mathcal{D}^{b} \xi^{a} + \left(
      T^{abc} - T^{bac}
    \right)
    \xi_{c}
  \right],
  \\
  \mathfrak{L}_{\xi}^{+} \phi & = \mathrm{I}_{\xi} \mathrm{d} \phi.
\end{align}
This is the set of transformations we will use to generalize the standard gauge fixing of GWs.


\subsection{Lie draggings vs. gauge transformations}
\label{sec:rant}

Before moving on, we would like to point out a common misunderstanding regarding the interpretation of an infinitesimal LD as a harmless ``gauge transformation'' on the fields. First, the Lagrangian, the vierbein, the spin connection, and the scalar field, although invariant under diffeomorphisms by virtue of being differential forms, transform nontrivially under infinitesimal LDs.
Second, diffeomorphism invariance is not a gauge symmetry in the sense that there is no principal bundle involved (in sharp contrast to the local Lorentz symmetry). Third, as shown in Sec.~\ref{sec:LCLD}, the Lie derivative of a Lorentz-tensor $p$-form does not transform covariantly under LLTs. Therefore, it is certainly more suitable to take the LDs and LLTs as the fundamental symmetries of the theory in the first-order formalism.  

Given a well-behaved theory for a field $\psi$, with field equations written symbolically as
$\mathcal{E} \left( \psi \right) = 0$,
we have that
\begin{equation}
  \mathcal{E} \left( \psi + \mathfrak{L}_{\xi} \psi\right) =
  \mathcal{E} \left( \psi \right) +
  \mathfrak{L}_{\xi} \mathcal{E} \left( \psi \right) = 0.
\end{equation}
This means that it is possible to map in an invertible way an on-shell configuration into a different one satisfying some practical condition we are interested in: the gauge fixing. Thus, solving the field equations for the latter is equivalent to solving them for the former, and it is only in this restricted sense that an LD can be identified with a gauge transformation.


\subsection{Perturbations on a Riemann--Cartan geometry}
\label{sec:Perturbations}

Infinitesimal LDs act on perturbations on the RC geometry in a way similar to the standard Riemannian case.
In Ref.~\cite{Izaurieta:2019dix}, it was shown that they can be described up to second order through the perturbations in the vierbein and the spin connection given by
\begin{align}
  e^{a} \mapsto \bar{e}^{a} & =
  e^{a} + \frac{1}{2} H^{a},
  \label{Eq_Perturb_Vierbein} \\
  \omega^{ab} \mapsto \bar{\omega}^{ab} & =
  \omega^{ab} + U^{ab} \left( H, \partial H \right) + V^{ab}.
  \label{Eq_Perturb_Spin_Connection}
\end{align}
Here, $H^{a} = \nwse{H}{a}{\mu} \mathrm{d}x^{\mu}$ is a one-form describing the vierbein perturbation, which is related to the canonical metric perturbation
$g_{\mu\nu} \mapsto g_{\mu\nu} + h_{\mu\nu}$ through
\begin{align}
  \nwse{H}{a}{\mu} & = \nwse{e}{a}{\rho} \left(
    \nwse{h}{\rho}{\nu} - \frac{1}{4} \nwse{h}{\rho}{\mu} \nwse{h}{\mu}{\nu} +
    \frac{1}{8} \nwse{h}{\rho}{\lambda} \nwse{h}{\lambda}{\mu} \nwse{h}{\mu}{\nu} + \cdots
  \right), \\
  h_{\mu\nu} & = \left(
    \nwse{e}{a}{\mu} + \frac{1}{4} \nwse{H}{a}{\mu}
  \right) H_{a\nu}.
\end{align}
Without loss of generality, its orthonormal-frame components can be taken as symmetric, $H_{ba} = H_{ab}$ for any Lorentz-invariant theory~\cite{Izaurieta:2019dix}.

The perturbation on the spin connection comes in two pieces, $U_{ab} \left( H, \partial H \right)  $ and $V_{ab}$.
The one-form $U_{ab}$ can be written in terms of $H^{a}$ as $U_{ab} = U_{ab}^{\left( 1 \right)} + U_{ab}^{\left( 2 \right)} + \mathcal{O} \left(  H^{3}\right)$, where
\begin{align}
  U_{ab}^{\left( 1 \right)} & = -\frac{1}{2} \left(
    \mathrm{I}_{a} \mathrm{D} H_{b} -
    \mathrm{I}_{b} \mathrm{D} H_{a}
  \right),
  \label{Eq_Perturb_U1}
  \\
  U_{ab}^{\left( 2 \right)} & =
  \frac{1}{8} \mathrm{I}_{ab} \left(
    \mathrm{D} H_{c} \wedge H^{c}
  \right)
   \nonumber \\ &\quad
  - \frac{1}{2} \left[
    \mathrm{I}_{a} \left( U_{bc}^{\left( 1 \right)} \wedge H^{c} \right) -
    \mathrm{I}_{b} \left( U_{ac}^{\left( 1 \right)} \wedge H^{c} \right)
  \right].
  \label{Eq_Perturb_U2}
\end{align}
The one-form $V_{ab}$ corresponds to a purely torsional perturbation mode, independent of $H^{a}$.

In terms of the perturbations, torsion and curvature behave as
\begin{align}
  T_{a} & \mapsto \bar{T}_{a} = T_{a} +
  T_{a}^{\left( 1 \right)} + T_{a}^{\left( 2 \right)}, \\
  R^{ab} & \mapsto \bar{R}^{ab} = R^{ab} +
  R_{\left( 1 \right)}^{ab} + R_{\left( 2 \right)}^{ab},
\end{align}
with
\begin{align}
  T_{a}^{\left( 1 \right)} & =
  V_{ab} \wedge e^{b} - \frac{1}{2} \mathrm{I}_{a} \left( H^{b} \wedge T_{b} \right),
  \label{Eq_Perturb_T1} \\
  T_{a}^{\left( 2 \right)} & =
  \frac{1}{2} V_{ab} \wedge H^{b} +
  \frac{1}{4} \mathrm{I}_{a} \left[
    H^{c} \wedge \mathrm{I}_{c} \left( H^{b} \wedge T_{b} \right)
  \right],
  \label{Eq_Perturb_T2} \\
  R_{\left( 1 \right)}^{ab} & =
  \mathrm{D} U_{\left( 1 \right)}^{ab} + \mathrm{D} V^{ab},
  \label{Eq_Perturb_R1}
  \\
  R_{\left( 2 \right)}^{ab} & =
  \mathrm{D} U_{\left( 2 \right)}^{ab} +
  \left( U^{a}_{\left( 1 \right) c} + \nwse{V}{a}{c} \right) \wedge
  \left( U_{\left( 1 \right)}^{cb} + V^{cb} \right).
  \label{Eq_Perturb_R2}
\end{align}

We can always perform an MLD on the background and a GW perturbation simultaneously,
\begin{equation}
  e^{a} \mapsto \bar{e}^{a} = e^{a} +
  \mathfrak{L}_{\xi}^{+} e^{a} + \frac{1}{2} H^{a},
\end{equation}
defining a new GW as
\begin{equation}
  \frac{1}{2} H'^{a} = \frac{1}{2} H^{a} +
  \mathfrak{L}_{\xi}^{+} e^{a},
\end{equation}
and therefore we have what we might call the ``gauge transformation,''
\begin{equation}
  H_{a} \mapsto H'_{a} = H_{a} +
  \mathring{\mathrm{D}} \xi_{a} + \mathring{\mathcal{D}}_{a} \xi.
\end{equation}
It is possible to use this relation to prove that
\begin{equation}
  \mathcal{D}^{a} H'_{a} - \frac{1}{2} \mathrm{d} H' =
  \mathcal{D}^{a} \mathring{\mathcal{D}}_{a} \xi +
  \mathrm{I}^{a} \mathrm{D} \mathring{\mathrm{D}} \mathrm{I}_{a} \xi +
  \mathcal{D}^{a} H_{a} - \frac{1}{2} \mathrm{d} H,
\end{equation}
where $\xi = \xi_{\mu} \mathrm{d}x^{\mu}$ and $H = \nwse{H}{a}{a}$.
Since it is always possible to find a one-form field $\xi = \xi_{\mu} \mathrm{d}x^{\mu}$ satisfying
$\mathcal{D}^{a} \mathring{\mathcal{D}}_{a} \xi +
 \mathrm{I}^{a} \mathrm{D} \mathring{\mathrm{D}} \mathrm{I}_{a} \xi +
 \mathcal{D}^{a} H_{a} - \frac{1}{2} \mathrm{d} H = 0$,
we can always choose $H_{a}$ such that
\begin{equation}
  \mathcal{D}_{a} H^{a} - \frac{1}{2} \mathrm{d} H = 0.
  \label{Eq_Lorenz_Gauge}
\end{equation}
This is the RC-geometry generalization of the standard Lorenz gauge fixing $\mathring{\nabla}^\mu\bar{h}_{\mu\nu} = 0$ on the trace-reversed variable $\bar{h}_{\mu\nu} = h_{\mu\nu} - \frac{1}{2}g_{\mu\nu}h$ of standard Riemannian geometry.

In the following sections, we will use the mathematical tools we have developed in Sec.~\ref{Sec_Math} and the condition~(\ref{Eq_Lorenz_Gauge}) to study the propagation of GWs in the nonminimal GB-coupling case.


\section{Sizes and Frequencies}
\label{Sec_Scales}

Even in the standard torsionless case, it is a nontrivial task to separate GWs from the background geometry.
In general, one must consider an expansion in two kind of variables: amplitudes and frequencies.
A GW is well defined only in the regime when small and rapidly changing perturbations move over a slowly varying background.
We follow the approach of Ref.~\cite[Ch.~1.5]{Maggiore:1900zz} as closely as possible, but considering a nonvanishing torsion.

Let us normalize the analysis by choosing a vierbein of components
\begin{equation}
  \left\vert \nwse{e}{a}{\mu} \right\vert \sim 1,
\end{equation}
describing a slowly changing geometry over a characteristic length scale $L$.
Since the torsion two-form is given by
$T^{a} = \mathrm{d} e^{a} + \nwse{\omega}{a}{b} \wedge e^{b}$,
we infer that both torsion and the spin connection must be of magnitude 
\begin{align}
  \left\vert T^{a} \right\vert & \sim \frac{1}{L}, &
  \left\vert \omega^{ab} \right\vert & \sim \frac{1}{L},
\end{align}
while the Lorentz curvature turns out to be of magnitude
\begin{equation}
  \left \vert R^{ab} \right \vert \sim \frac{1}{L^{2}}.
\end{equation}

Let us label the amplitude scales for perturbations as%
\footnote{Later on we use $H$ to denote the trace $\nwse{H}{a}{a}$, which is of course unrelated to $H$ as the scale for metric perturbations. We can only hope that the reader will be able to tell one from the other according to context.}
\begin{align}
  \left\vert H^{a}  \right\vert & \sim H, &
  \left\vert V^{ab} \right\vert & \sim V,
\end{align}
with $H \ll 1$ and $V \ll 1$.
These perturbations change rapidly in the wavelength scales $\lambdabar_{H}$ and $\lambdabar_{V}$,
\begin{align}
  \left\vert \partial H^{a}  \right\vert & \sim \frac{H}{\lambdabar_{H}}, &
  \left\vert \partial V^{ab} \right\vert & \sim \frac{V}{\lambdabar_{V}}.
\end{align}
These wavelengths are small compared with the scale $L$ of the background geometry,
\begin{align}
  \epsilon_{H} & = \frac{\lambdabar_{H}}{L} \ll 1, &
  \epsilon_{V} & = \frac{\lambdabar_{V}}{L} \ll 1.
\end{align}

In order to relate the perturbation scales $H$ and $V$, let us observe that the torsion components are $1/L$ times smaller than those of the vierbein.
Since the perturbed torsion is given by
[cf.~Eqs.~(\ref{Eq_Perturb_T1})--(\ref{Eq_Perturb_T2})]
$\bar{T}_{a} = T_{a} + T_{a}^{\left(  1 \right)} + T_{a}^{\left( 2 \right)} + \mathcal{O} \left( H^{3} \right)$,
it is natural to expect $T_{a}^{\left( 1 \right)}$ and $T_{a}^{\left( 2 \right)}$ to be also $1/L$ times smaller than the vierbein perturbations,
\begin{align}
\left\vert T_{a}^{\left( 1 \right)} \right\vert & \sim \frac{H}{L}, &
\left\vert T_{a}^{\left( 2 \right)} \right\vert & \sim \frac{H^{2}}{L}.
\end{align}
This is the same as requiring the perturbation scales of $V^{ab}$ and $H^{a}$ to be related by
\begin{equation}
  V \sim \frac{H}{L},
  \label{Eq_V=H_sobre_L}
\end{equation}
meaning that the torsional modes are much weaker than the metric ones.

The curvature perturbations (\ref{Eq_Perturb_R1})--(\ref{Eq_Perturb_R2}) include $1/\epsilon^{2}$ and $1/\epsilon$ terms on each order,
\begin{align}
  \left\vert \mathrm{D} U_{\left( 1 \right)}^{ab} \right\vert & \sim
  \frac{1}{L^{2}} \frac{H}{\epsilon_{H}^{2}}, \\
  \left\vert \mathrm{D} V^{ab} \right\vert & \sim
  \frac{1}{L^{2}} \frac{H}{\epsilon_{V}}, \\
  \left\vert
    \mathrm{D} U_{\left( 2 \right)}^{ab} +
    U^{a}_{\left( 1 \right) c} \wedge U_{\left( 1 \right)}^{cb}
  \right\vert & \sim
  \frac{1}{L^{2}} \frac{H^{2}}{\epsilon_{H}^{2}}, \\
  \left\vert
    U^{a}_{\left( 1 \right) c} \wedge V^{cb} +
    \nwse{V}{a}{c} \wedge U_{\left( 1 \right)}^{cb}
  \right\vert & \sim
  \frac{1}{L^{2}} \frac{H^{2}}{\epsilon_{H}}, \\
  \left\vert \nwse{V}{a}{c} \wedge V^{cb} \right\vert & \sim
  \frac{H^{2}}{L^{2}}.
\end{align}

At this point, since Eq.~(\ref{Eq_EOM_Vierbein_EH}) has the same form as the canonical Einstein--Hilbert equations, and since the leading terms in the expansions are the metric modes, much of the analysis goes exactly as in the standard case.
The field equations must be split into low- and high-frequency pieces.
From the low-frequency piece it is straightforward to prove that
\begin{equation}
  H \ll \epsilon_{H} \ll 1,
  \label{Eq_H comp_epsilon_comp_1}
\end{equation}
and taking this into consideration, the high-frequency piece of Eq.~(\ref{Eq_EOM_Vierbein_EH}) to leading and subleading orders corresponds just to
\begin{equation}
  \epsilon_{abcd} R_{\left( 1 \right)}^{ab} \wedge e^{c} = 0.
  \label{Eq_High_Frequency}
\end{equation}
In Sec.~\ref{Sec_GW_Speed}, we analyze the behavior of GWs and torsional modes predicted by Eq.~(\ref{Eq_High_Frequency}).


\section{The dispersion relation, speed and polarization of gravitational waves}
\label{Sec_GW_Speed}

\subsection{The wave equation and torsional obstruction to the transverse-traceless gauge}

The left-hand side of Eq.~\eqref{Eq_High_Frequency} can be written as 
(see Appendix~\ref{Apendice} for the algebraic details)
 \begin{equation}
  \epsilon_{abcn} R_{\left( 1 \right)}^{ab} \wedge e^{c} =
  \left( \mathrm{I}_{n} W_{m} - \frac{1}{2} \eta_{mn} \mathrm{I}_{p} W^{p} \right) \hodge e^{m}, \label{Eq_Pert1=W-W/2}
\end{equation}
where
\begin{equation}
  W_{m} = -\mathcal{D}_{a} \mathcal{D}^{a} H_{m} +
  \left[ \mathcal{D}_{a}, \mathcal{D}_{m} \right] H^{a} +
  2\mathrm{I}_{a} \mathrm{D} V^{a}{}_{m}.
\end{equation}
It is clear that Eq.~(\ref{Eq_High_Frequency}) and Eq.~(\ref{Eq_Pert1=W-W/2}) imply $W_{m}=0$ as the equation for GWs, and therefore from now on the equation
\begin{equation}
  -\mathcal{D}_{a} \mathcal{D}^{a} H_{m} +
  \left[ \mathcal{D}_{a}, \mathcal{D}_{m} \right] H^{a} +
  2\mathrm{I}_{a} \mathrm{D} V^{a}{}_{m} = 0,
  \label{Eq_GW_Inhomo}
\end{equation}
will be the protagonist of our analysis.

As in the torsionless case, the commutator
$\left[ \mathcal{D}_{a}, \mathcal{D}_{m} \right] H^{a}$
gives rise to some inhomogeneous terms with curvature and torsion,
\begin{align}
  \left[ \mathcal{D}_{a}, \mathcal{D}_{m} \right] H^{a} & =
  \mathrm{I}_{am} \left( R^{a}{}_{b} \wedge H^{b} \right) 
  \nonumber \\ &\quad +
  \mathrm{I}_{a} \left( R^{ab} H_{bm} + R_{mb} H^{ab} \right) 
  \nonumber \\ &\quad -
  \left( \mathrm{D} T_{abm} H^{ab} + T_{abm} \mathcal{D}^{a} H^{b} \right).
  \label{Eq_AntiCommH}
\end{align}
Replacing these inhomogeneous terms and using the generalized Weitzenböck identity of Lemma~\ref{lem:GenWeitzen} (see also Ref.~\cite{Barrientos:2019msu}), $W_{m}$ can be written in terms of the generalized de~Rham--Laplace wave operator (cf.~Definition~\ref{def:GdRL}) as
\begin{align}
  W_{m} & = \blacksquare_{\mathrm{dR}} H_{m} +
  \mathrm{I}_{am} \left( R^{a}{}_{b} \wedge H^{b} \right) 
  \nonumber \\ &\quad -
  \left( \mathrm{D} T_{abm} H^{ab} + T_{abm} \mathcal{D}^{a} H^{b} \right) +
  2\mathrm{I}_{a} \mathrm{D} V^{a}{}_{m}. \label{Eq_Wm}
\end{align}

This implies that the analysis of Ref.~\cite{Barrientos:2019msu} has to be generalized to include the extra inhomogeneous terms in Eq.~(\ref{Eq_Wm}).
A second important observation regarding $W_{m}$ is that its ``trace,'' $\mathrm{I}_{p} W^{p}$, consists only of torsional inhomogeneous terms besides the wave operator acting on $H=H^{a}{}_{a}$, 
\begin{align}
  \mathrm{I}_{p} W^{p} & = -\mathcal{D}_{a} \mathcal{D}^{a} H +
  T^{abc} \left( \mathcal{D}_{c} H_{ab} + 2 T_{bcd} H^{d}{}_{a} \right) 
  \nonumber \\ &\quad +
  2\mathrm{I}_{ab} \mathrm{D} V^{ab}.
  \label{Eq_Trace_W}
\end{align}

In general lines, the metric mode of GWs behaves similarly to the standard torsionless case, albeit with an important and subtle difference.
Let us observe that, in the standard case, besides the Lorenz gauge fixing~(\ref{eq:LGF}),
it is possible to perform (in a vacuum region only) an additional gauge transformation to render $h_{\mu\nu}$ traceless: the \emph{transverse-traceless gauge}. In our case, things are a bit more complicated.
Under an infinitesimal LD generated by $\tilde{\xi}$, the trace of the metric perturbation,
$H = \nwse{H}{a}{a}$, changes as
\begin{equation}
  H \mapsto \tilde{H} = H + 2 \mathring{\mathcal{D}}_{a} \tilde{\xi}^{a}.
\end{equation}
Naively it may seem possible to choose a $\tilde{\xi}$ such that $\tilde{H}=0$, as long as $\tilde{\xi}$ also satisfies $\mathcal{D}^{a}\mathcal{\mathring{D}}_{a}\tilde{\xi}+\mathrm{I}^{a}\mathrm{D\mathring{D}I}_{a}\tilde{\xi}=0$ in order not to spoil the Lorenz condition~(\ref{Eq_Lorenz_Gauge}). 
The problem with such a construction is Eq.~(\ref{Eq_Trace_W}).
A traceless $H^{a}$ would create an unphysical constraint between the perturbations and torsion,
\begin{equation}
  T^{abc} \left(
    \mathcal{D}_{c} H_{ab} + 2 T_{bcd} H^{d}{}_{a}
  \right)
  + 2 \mathrm{I}_{ab} \mathrm{D} V^{ab} = 0,
  \label{Eq_metida_pata}
\end{equation}
and therefore, in general we must have $H \neq 0$.\footnote{In the standard GR torsionless case, the constraint in Eq.~(\ref{Eq_metida_pata}) vanishes identically and it is possible to impose $h=0$ in a vacuum region.
Some further conditions, such as $h_{0\mu}=0$, are only possible on a flat background even in the torsionless case.}
Torsion thus creates an obstruction to the popular transverse-traceless gauge.
However, comparing orders of magnitude in the terms of Eq.~(\ref{Eq_Trace_W}), we conclude that using a wisely chosen LD we may get a trace that is $\epsilon_{H}$ times smaller than the typical magnitude of the metric perturbation.

In the next section, we will use the eikonal limit of GWs to obtain valuable information.
For instance, from simple inspection of Eq.~(\ref{Eq_Wm}) we can see that the dispersion relation for $H_{m}$ will not be modified at leading order by torsion.
This means that $H_{m}$ propagates at the speed of light on null geodesics, as in the standard torsionless case.
Even further, in Eq.~(\ref{Eq_AntiCommH}) the terms
$\mathrm{I}_{am} \left( R^{a}{}_{b} \wedge H^{b} \right)$,
$\mathrm{I}_{a} \left( R^{ab} H_{bm} + R_{mb} H^{ab} \right)$, and
$\mathrm{D} T_{abm} H^{ab}$
are all of order $H/L^{2}$ and irrelevant to the eikonal limit at both leading and subleading order.
However, the terms $T_{abm} \mathcal{D}^{a} H^{b}$ and $2\mathrm{I}_{a} \mathrm{D} V^{a}{}_{m}$ in Eq.~(\ref{Eq_Wm}) modify the propagation of GW polarization and the conservation of the ``number of rays'' at subleading order, generalizing the results of Ref.~\cite{Barrientos:2019msu} for fields obeying the homogeneous equation
$\left( \mathrm{D}^{\ddag} \mathrm{D} + \mathrm{DD}^{\ddag} \right) H_{m}=0$. In the next section we analyze these affirmations in detail.


\subsection{The eikonal limit of gravitational waves}

Let us write the vierbein and torsional perturbations $H^{a}$ and $V^{ab}$ as%
\footnote{The physical perturbations correspond to the real parts of these complex quantities.}
\begin{align}
  H^{a}  & = \euler^{i\theta} \mathcal{H}^{a}, &
  V^{ab} & = \euler^{i\theta} \mathcal{V}^{ab},
\end{align}
where $\theta$ is a rapidly changing (on a characteristic scale $\lambdabar$) real phase, which, for simplicity, we take to be the same for both geometrical modes.
Here, $\mathcal{H}^{a}=\mathcal{H}^{a}{}_{b}e^{b}$ and $\mathcal{V}^{ab}=\mathcal{V}^{ab}{}_{c}e^{c}$ correspond to slowly changing (on a characteristic scale $L$) one-forms with complex-valued components $\mathcal{H}^{a}{}_{b}$ and $\mathcal{V}^{ab}{}_{c}$.

In terms of the characteristic scales $\lambdabar$ and $L$ we
define the eikonal parameter $\epsilon = \lambdabar / L$.
From Eqs.~(\ref{Eq_V=H_sobre_L}) and~(\ref{Eq_H comp_epsilon_comp_1}) one can easily show that it satisfies
\begin{equation}
  V \ll H \ll \epsilon \ll 1.
\end{equation}
In terms of these, the transverse condition~(\ref{Eq_Lorenz_Gauge}) becomes
\begin{align}
  \mathcal{D}_{a} H^{a} - \frac{1}{2} \mathrm{d} H & =
  i \euler^{i\theta} \left(
    k_{a} \mathcal{H}^{a} - \frac{1}{2} k \mathcal{H}
  \right) 
  \nonumber \\ &\quad +
  \euler^{i\theta} \left(
    \mathcal{D}_{a} \mathcal{H}^{a} - \frac{1}{2} \mathrm{d} \mathcal{H}
  \right),
\end{align}
while $W_{m}$ corresponds to
\begin{align}
  W_{m} & = \euler^{i\theta} \bigg\{
    k_{a} k^{a} \mathcal{H}_{m} - 2i \left[
      k^{a} \left(
        \mathcal{D}_{a} \mathcal{H}_{m} +
        \frac{1}{2} T_{abm} \mathcal{H}^{b}
      \right) 
    \right.
   \nonumber \\ &\quad \left.
      -\mathrm{I}_{a} \left( k \wedge \mathcal{V}^{a}{}_{m} \right)
      +\frac{1}{2} \mathcal{H}_{m} \mathcal{D}_{a} k^{a}
    \right] 
  \nonumber \\ &\quad 
    -\mathcal{D}_{a} \mathcal{D}^{a} \mathcal{H}_{m} +
    \left[ \mathcal{D}_{a}, \mathcal{D}_{m} \right]
    \mathcal{H}^{a} + 2 \mathrm{I}_{a} \mathrm{D} \mathcal{V}^{a}{}_{m}
  \bigg\},
  \label{Eq_W_Eikonal}
\end{align}
where the wave one-form $k$ is given by
\begin{equation}
  k = \mathrm{d} \theta = k_{a} e^{a} = k_{\mu} \mathrm{d} x^{\mu}.
\end{equation}
Equation~(\ref{Eq_W_Eikonal}) generalizes Eq.~(61) of Ref.~\cite{Barrientos:2019msu} to take into account the inhomogeneous terms in Eq.~(\ref{Eq_Wm}).

We can expand $\mathcal{H}^{a}$ and $\mathcal{V}^{ab}$ as
\begin{align}
  \mathcal{H}^{a} & =
  \sum_{n=0}^{\infty} \mathcal{H}_{\left( n \right)}^{a}, &
  \mathcal{V}^{ab} & =
  \sum_{n=0}^{\infty} \mathcal{V}_{\left( n \right)}^{ab},
\end{align}
where
$\mathcal{H}_{\left( n \right)}^{a}$ and
$\mathcal{V}_{\left( n \right)}^{ab}$ are of order $\epsilon^{n}$.
The leading orders,
$\mathcal{H}_{\left( 0 \right)}^{a}$ and $\mathcal{V}_{\left( 0 \right)}^{ab}$,
correspond to dominant, $\lambdabar$-independent pieces.
The subleading order, $n=1$, describes the propagation of polarization, while terms with $n \geq 2$ correspond to higher-order deviations from the geometric optics limit.


To leading order, the equation $W_{m}=0$ implies the canonical dispersion relation for the metric mode,
\begin{equation}
  k_{a} k^{a} = k_{\mu}k^{\mu} = 0,
  \label{Eq_Leading_Dispersion}
\end{equation}
and the standard transverse condition,
\begin{equation}
  k_{a} \mathcal{H}_{\left( 0 \right)}^{a} - \frac{1}{2} k \mathcal{H}_{\left( 0 \right)} = 0.
\end{equation}
This means that, to leading order in the dispersion relation, there is no difference with the standard GR torsionless case. However, at subleading order, $W_{m}=0$ gives rise to new interactions with torsion,
\begin{align}
 0 &= k^{a} \left(
    \mathcal{D}_{a} \mathcal{H}_{m}^{\left( 0 \right)} +
    \frac{1}{2} T_{abm} \mathcal{H}_{\left( 0 \right)}^{b}
  \right) \notag \\ &\quad -
  \mathrm{I}_{a} \left(
    k \wedge \mathcal{V}_{\left( 0 \right)}^{a}{}_{m}
  \right) +
  \frac{1}{2} \mathcal{H}_{m}^{\left( 0 \right)} \mathcal{D}_{a} k^{a} ,
  \label{Eq_Subleading_Amplitude_Prop}
\end{align}
and the transverse condition assumes the form
\begin{equation}
  \mathcal{D}_{a} \mathcal{H}_{\left( 0 \right)}^{a} -
  \frac{1}{2} \mathrm{d} \mathcal{H}_{\left( 0 \right)} = 0.
\end{equation}

It may seem strange to have the usual dispersion relation~(\ref{Eq_Leading_Dispersion}) even in this case.
On a geometry with nonvanishing torsion, one may expect GWs to propagate on null auto-parallels.\footnote{It is worth noticing that, when torsion is present, geodesics and auto-parallels do not necessarily coincide: while the former are curves of extremal length with respect to the metric, the latter are curves over which a vector is parallel transported with respect to itself according to the connection (see Ref.~\cite{Hehl76} for a discussion). }
However, this is not the case:  Eq.~(\ref{Eq_Leading_Dispersion}) implies that GWs propagate on null geodesics, regardless of the background torsion.

Differentiating Eq.~(\ref{Eq_Leading_Dispersion}) and using $k_{a}=\mathcal{D}_{a}\theta$, we find
\begin{equation}
  k^{a} \left( \mathcal{D}_{a} k_{b} + T^{c}{}_{ab} k_{c} \right) = 0.
\end{equation}
This result is equivalent to $k^{a} \mathring{\mathcal{D}}_{a} k^{b} = 0$,
and since $\mathring{\mathcal{D}}_{a} = \mathring{\nabla}_{a}$, we get
\begin{equation}
  k^{a} \mathring{\nabla}_{a} k^{b} = 0.
\end{equation}
This means that the metric mode of GWs travels along \emph{null geodesics}, not \emph{null auto-parallels}.


At subleading order, torsion gives rise to an anomalous propagation of polarization.
In order to analyze it, let us parametrize the wave polarization and amplitude as
$\mathcal{H}_{\left( 0 \right)}^{a}$ and $\mathcal{V}_{\left( 0 \right)}^{ab}$ through
\begin{align}
  \mathcal{H}_{\left( 0 \right)}^{a}  & = \mathcal{H} P^{a},
  \label{Eq_H=HP} \\
  \mathcal{V}_{\left( 0 \right)}^{ab} & = \mathcal{V} Q^{ab}.
  \label{Eq_V=VQ}
\end{align}
Here, wave polarization is described by the one-forms
$P^{a} = P^{a}{}_{c} e^{c}$ and
$Q^{ab} = Q^{ab}{}_{c} e^{c}$,
while wave amplitude is described by
$\mathcal{H}$ and $\mathcal{V}$.
The polarization components are complex valued,
$P^{a}{}_{c}, Q^{ab}{}_{c} \in \mathbb{C}$,
while the amplitude scalars are real,
$\mathcal{H}, \mathcal{V} \in \mathbb{R}$.
The polarization forms are normalized as
\begin{align}
  \left( \bar{P}_{a}  \wedge \hodge P^{a}  \right) & = v_{\left( 4 \right)}, \\
  \left( \bar{Q}_{ab} \wedge \hodge Q^{ab} \right) & = v_{\left( 4 \right)},
\end{align}
where $v_{\left( 4 \right)} = \frac{1}{4!} \epsilon_{abcd} e^{a} \wedge e^{b} \wedge e^{c} \wedge e^{d}$ is the volume four-form and a bar above a quantity denotes its complex conjugate.
This implies the following normalization on the wave polarization and amplitude:
\begin{align}
  \hodge \left( \bar{\mathcal{H}}_{a}^{\left( 0 \right)} \wedge
  \hodge \mathcal{H}_{\left( 0 \right)}^{a} \right) & = -\mathcal{H}^{2},
  \label{Eq_hodgeHH=-H2} \\
  \hodge \left( \bar{\mathcal{V}}_{ab}^{\left( 0 \right)} \wedge
  \hodge \mathcal{V}_{\left( 0 \right)}^{ab} \right) & = -\mathcal{V}^{2}.
  \label{Eq_hodgeVV=-V2}
\end{align}


Let $J = \mathcal{H}^{2} k$ be the ``number-of-rays'' current density one-form.
When considering the eikonal limit and standard Riemannian geometry, this form is conserved,
$\mathrm{d}^{\dag} J = 0$.
However, this is no longer true for a geometry with nonvanishing torsion, as was first shown in Ref.~\cite{Barrientos:2019msu} for the homogeneous wave equation case.
The current situation is similar, but new terms have to be added.
Since $\mathrm{d}^{\dag} J = T_{abc} \eta^{ab} J^{c} - \mathcal{D}_{a} J^{a}$, let us start by computing
\begin{equation}
  \mathcal{D}_{a} J^{a} = \mathcal{D}_{a} \left( \mathcal{H}^{2} \right) k^{a} +
  \mathcal{H}^{2} \mathcal{D}_{a} k^{a}.
\end{equation}
Using Eq.~(\ref{Eq_hodgeHH=-H2}), Lemma~\ref{lem:Com_D_Hodge} and Eq.~(\ref{Eq_H=HP}), we have that
\begin{align}
  \mathcal{D}_{a} \left( \mathcal{H}^{2} \right) & =
  -\hodge \left( \mathcal{D}_{a} \bar{\mathcal{H}}_{c} \wedge \hodge \mathcal{H}^{c} +
  \mathcal{D}_{a} \mathcal{H}_{c} \wedge \hodge \bar{\mathcal{H}}^{c} \right) 
  \nonumber \\ &\quad -
  \mathcal{H}^{2} T_{bca} \Pi^{bc},
\end{align}
where
\begin{equation}
  \Pi^{ab} = \eta^{ab} - \frac{1}{2} \left(
    \bar{P}^{ca} P_{c}{}^{b} + P^{ca} \bar{P}_{c}{}^{b}
  \right).
\end{equation}
This allows us to write
\begin{align}
  \mathcal{D}_{a} J^{a} & = -\hodge \left(
    k^{a} \mathcal{D}_{a} \bar{\mathcal{H}}_{c} \wedge \hodge \mathcal{H}^{c} +
    k^{a} \mathcal{D}_{a} \mathcal{H}_{c} \wedge \hodge \bar{\mathcal{H}}^{c}
  \right) 
  \nonumber \\ &\quad -
  T_{abc} \Pi^{ab} J^{c} + \mathcal{H}^{2} \mathcal{D}_{a} k^{a}.
\end{align}
Using Eq.~(\ref{Eq_Subleading_Amplitude_Prop}), we find
\begin{equation}
  \mathcal{D}_{a} J^{a} = \left(
    \frac{\mathcal{V}}{\mathcal{H}} \Theta_{c} - T_{abc} \Pi^{ab}
  \right) J^{c},
  \label{Eq_DaJa}
\end{equation}
with
\begin{equation}
  \Theta_{c} = \left(
    \bar{Q}_{c}{}^{pq} P_{pq} + \bar{Q}^{ba}{}_{a} P_{bc}
  \right) + \left(
    Q_{c}{}^{pq} \bar{P}_{pq} + Q^{ba}{}_{a} \bar{P}_{bc}
  \right).
\end{equation}
From here, we can see that $J$ is no longer conserved in the eikonal limit, that is
\begin{equation}
  \mathrm{d}^{\dag} J = \left[
    T_{abc} \left( \eta^{ab} + \Pi^{ab} \right) -
    \frac{\mathcal{V}}{\mathcal{H}} \Theta_{c}
  \right]
  J^{c}.
  \label{Eq_dJ_no_consv}
\end{equation}
Observe that, even on backgrounds with vanishing torsion, a nonzero torsional perturbation will lead to a breakdown in the number-of-rays conservation, $\mathrm{d}^{\dag} J \neq 0$.

Another consequence of torsion is that the polarization one-form $P^{a}$ is no longer parallel transported along the GW trajectory.
This can be seen by using Eqs.~(\ref{Eq_Subleading_Amplitude_Prop}) and~(\ref{Eq_DaJa}), which allow us to write 
\begin{align}\notag
 0 &=  k^{a} \mathcal{D}_{a} P_{m} + \left(
    T_{cbm} P^{b} - \frac{1}{2} T_{abc} \Pi^{ab} P_{m}
  \right) k^{c} \\ &\quad +
  \frac{\mathcal{V}}{\mathcal{H}} \left[
    \frac{1}{2} \Theta_{c} k^{c} P_{m} -
    \mathrm{I}_{a} \left( k \wedge Q^{a}{}_{m} \right)
  \right],
\end{align}
which in component language reads
\begin{align}
  k^{c} \mathring{\nabla}_{c} P_{mn} & =
  \frac{1}{2} k^{c} \left[
    T_{mpc} P^{p}{}_{n} + T_{npc}P^{p}{}_{m} + T_{abc} \Pi^{ab} P_{mn} 
  \right.
  \nonumber \\ &\quad +
    \frac{\mathcal{V}}{\mathcal{H}} \left(
      Q_{cmn} + Q_{cnm} + \eta_{mc} Q_{na}{}^{a} 
    \right. \nonumber \\ &\quad + \left. \left.
      \eta_{nc} Q_{ma}{}^{a} - \Theta_{c} P_{mn}
    \right)
  \right].
  \label{Eq_Prog_Polarization_Comp}
\end{align}
From Eq.~(\ref{Eq_Prog_Polarization_Comp}) we observe that the polarization components $P_{mn}$ are parallel transported along the trajectory of the GW (i.e., the relation $k^{c}\mathring{\nabla}_{c} P_{mn} = 0$ is fulfilled) only in the Riemannian geometry case, when both background torsion and its perturbations vanish, $T_{abc}=0$, $V^{ab}=0$.
In the general case, the propagation of polarization will be disturbed along the GW trajectory by interactions with background torsion and ``rotons.''

To summarize, torsion does not change the dispersion relation~(\ref{Eq_Leading_Dispersion}), and therefore it changes neither the speed nor the direction of propagation of GWs.
At leading order in the eikonal approximation, its propagation remains the same as in standard GR.
The only difference lies in how torsion changes the propagation of polarization along the GW trajectory [cf.~Eq.~(\ref{Eq_Prog_Polarization_Comp})].


\section{Conclusions and future possibilities}
\label{Sec_TheEnd}

Observational data coming from multimessenger astronomy indicate, to a very high precision, that GWs travel at the speed of light. Due to this, several scalar-tensor theories that predict an anomalous GW speed have been dramatically constrained, with the scalar-GB coupling a particular example of this class. In the present article we showed that dispensing with the torsionless condition---usually \emph{assumed} in gravitational theories---allows the latter case to be reconciled with observations, in stark contrast with the usual torsionless case.


\begin{figure*}
  \centering
  \begin{tikzpicture}[
    x = \textwidth/3.2,
    y = \textwidth/3.2,
    every shadow/.style = {
      shadow xshift = 0pt,
      shadow yshift = 0pt,
      shadow opacity = 20,
    }
  ]
    \foreach \cp / \mx / \my / \thr / \th / \r in {
      13 /  90 / 50 / -20 /  90 / 0.7,
      11 / 100 / 80 /  80 / -80 / 0.4,
      11 /  50 / 30 / 110 / 240 / 0.7,
      11 /  70 / 30 /   0 / 200 / 0.6,
      11 /  30 / 30 /   0 / 180 / 0.2,
      11 /  30 / 30 /   0 / 135 / 0.8,
      15 / 100 / 20 /  10 / 160 / 0.6,
      11 /  80 / 50 / -10 /  15 / 0.4,
      11 /  30 / 30 /   0 /  45 / 0.7,
      11 /  30 / 30 /   0 / -20 / 0.7
    }
    {
      \node [cloud, blur shadow, cloud puffs = \cp,
             minimum width = \mx pt, minimum height = \my pt, rotate = \thr]
        at (\th:\r) {};
    }
    \def\nlisa{24}
    \foreach \k in {1, 2, ..., \nlisa} {
      \pgfmathsetmacro{\th}{\k*360/\nlisa}
      \draw (\th:1) pic [scale = 2/\nlisa, rotate = 120*rnd] {lisa};
    }
    \draw [<-] (10*360/\nlisa:1.12) -- +(135:0.1)
      node [anchor = south east] {GW detector};
    \draw [densely dashed, ->] (-1, -1) .. controls (135:0.2) .. (1, 1)
      node [pos = 0.05, anchor = south east] {Incoming GW, $P_{mn}^{\text{(in)}}$}
      node [pos = 0.95, anchor = north west] {Outgoing GW, $P_{mn}^{\text{(out)}}$}
      node [anchor = south west] {GW};
    \draw [densely dashed, ->] (-1, 1) .. controls (45:0.2) .. (1, -1)
      node [anchor = north west] {GW};
  \end{tikzpicture}
  \caption{Conceptual diagram of a possible way to map torsion ``clouds'' inside a region of space by measuring changes in the polarization of incoming/outgoing GWs. Each small triangle represents a GW detector.}
  \label{Fig_GW}
\end{figure*}
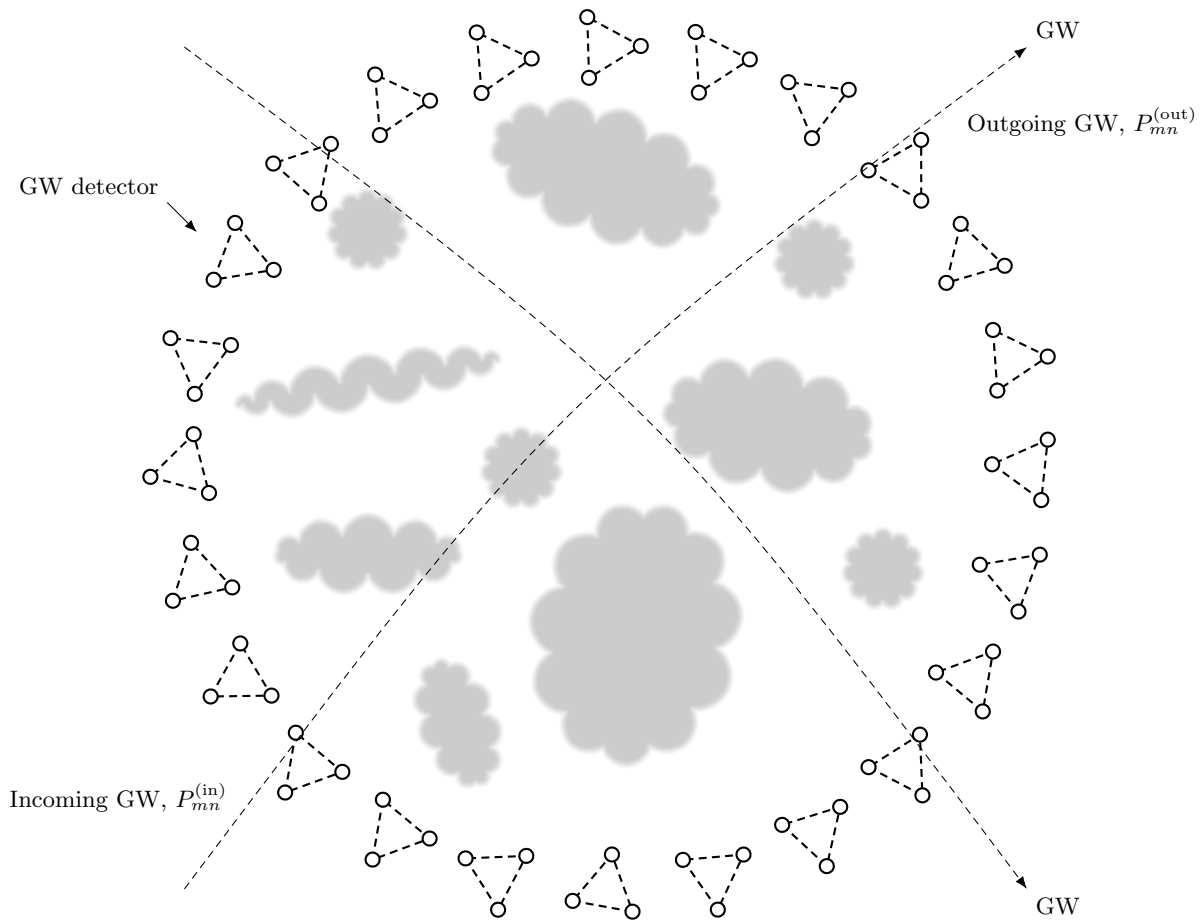


Our starting point is a fairly standard scalar-tensor gravity theory, defined by the Lagrangian~(\ref{Eq_Lagrangian-0}).
Crucially, removing the torsionless condition makes the vierbein and the spin connection independent degrees of freedom, and the same Lagrangian can give rise to two radically different dynamical theories, depending on whether or not this hypothesis is assumed \emph{a priori}.
One common misconception is to consider the torsional case as an exotic and small departure from the torsionless case.
That may be true for some of the ECSK phenomenology, but it is definitely not the case for large sectors of the Horndeski Lagrangian~\cite{Valdivia:2017sat}.
Another misconception is the belief that in order to recover the standard Riemannian geometry, it suffices to impose $T^{a}=0$ in all equations.
Both misconceptions arise from the failure to recognize that the torsionless condition is a strong constraint on the geometry, in the sense that we must add a Lagrange multiplier to impose it, as shown in Eqs.~(\ref{Eq_L-LM})--(\ref{Eq_L_Multiplier}) and Ref.~\cite{Valdivia:2017sat}. In this sense, the torsionless condition amounts to an unproven hypothesis.
Imposing it means \textit{adding} a hypothesis to the theory and a constraint to the geometry.

The dynamics derived from the field equations makes it plausible to expect that GWs propagate at the speed of light.\footnote{We notice that a similar behavior occurs in the first-order formulation of Chern--Simons modified gravity~\cite{Alexander:2008wi}, even though its metric formulation is compatible with the luminal propagation of GWs, as shown in Ref.~\cite{Nishizawa:2018srh}.}
To prove it, we have developed some new mathematical tools to study the wave operator on an RC geometry and a new approach to study perturbations on a background with dynamical torsion. Besides the standard metric mode, there is an additional one associated to the torsional degrees of freedom. We showed how to use the generalized Lie derivative to find a ``generalized Lorenz gauge fixing'' for the case of nonvanishing torsion.
These provide the necessary mathematical tools to tackle the problem of GWs on a dynamical torsion background, which allowed us to determine their speed in the current theory.

After a general analysis of sizes and frequencies, we found the inhomogeneous gravitational wave equation for this theory [cf.~Eq.~(\ref{Eq_GW_Inhomo})], and we discovered that torsion obstructs the popular transverse-traceless gauge. Then, we proceeded with the geometric optics (eikonal) approximation. To leading order, we recovered the canonical dispersion relation $k_{\mu} k^{\mu} = 0$, implying that GWs travel on null geodesics. The point is subtle but essential. EMWs always travel on null geodesics, regardless of any background torsion~\cite{Barrientos:2019msu}.
If GWs would have traveled on null auto-parallels, this might lead to an unobserved delay between GW/EMW, even if both of them were traveling at the same speed. Our results show that this is not the case: EMWs and GWs travel at the same speed and on the same kind of trajectory.

Torsion does affect GWs at subleading order, though.
In the torsionless case, both EMWs and GWs satisfy two eikonal limit conditions:
(i)~the number-of-rays current density is conserved, and
(ii)~polarization is parallel transported along the null geodesic trajectory.
In the current torsional context, these conditions are satisfied by EMWs but they are no longer valid for GWs: torsion breaks down the conservation of the number of rays and GW polarization is no longer parallel transported along the null geodesic. As discussed in Ref.~\cite{Barrientos:2019msu}, this behavior is generic of waves propagating on a torsional background, and not just a feature of the GB coupling.

This fact opens up many questions regarding how torsion modifies the propagation phenomenology of metric and torsional polarization, which lie beyond the scope of this paper.
The fact that GW polarization is affected by torsion means, however, that one may envision ways to use it to study the distribution of torsion.

Let us consider space-based GW detectors distributed on the surface of a vast spherical region, as suggested in Fig.~\ref{Fig_GW} (e.g., GW detectors orbiting the Sun at a distance of several AU at different angles with the ecliptic plane).
If each detector is capable of measuring GW polarization, the whole set of detectors could make a ``torsion tomography'' of sorts of the enclosed region, comparing the polarization components $P_{mn}$ of incoming and outgoing GWs.
For instance, the degree of rotation of GW polarization would encode information about torsion within the sphere. This technique may seem (and probably is) far fetched, given current technological possibilities.
However, from a certain point of view, it is also a very conservative idea.
One may regard it as the scaled-up gravitational version of Faraday's 1845 setup for measuring what is now known as Faraday rotation (see Fig.~\ref{Fig_Faraday}).
In this effect, certain transparent dielectric materials present circular birefringence when a magnetic field goes through them.
The different speeds of the two circular polarization modes have the effect of rotating the plane of polarization of light going through the dielectric parallel to the magnetic field, providing information on the properties of the material.
In the gravitational version, the change in the propagation of polarization is caused not by birefringence, since all GW polarization modes travel on null geodesics, but by the interaction with torsion. There exists, of course, the possibility of such an experiment indicating that torsion vanishes.
In such a case, the torsionless condition would become a valuable clue from nature on the kind of solutions of our theories that match reality rather than an untested hypothesis on their structure.


\begin{figure}
  \centering
  \includegraphics[width=.8\columnwidth]{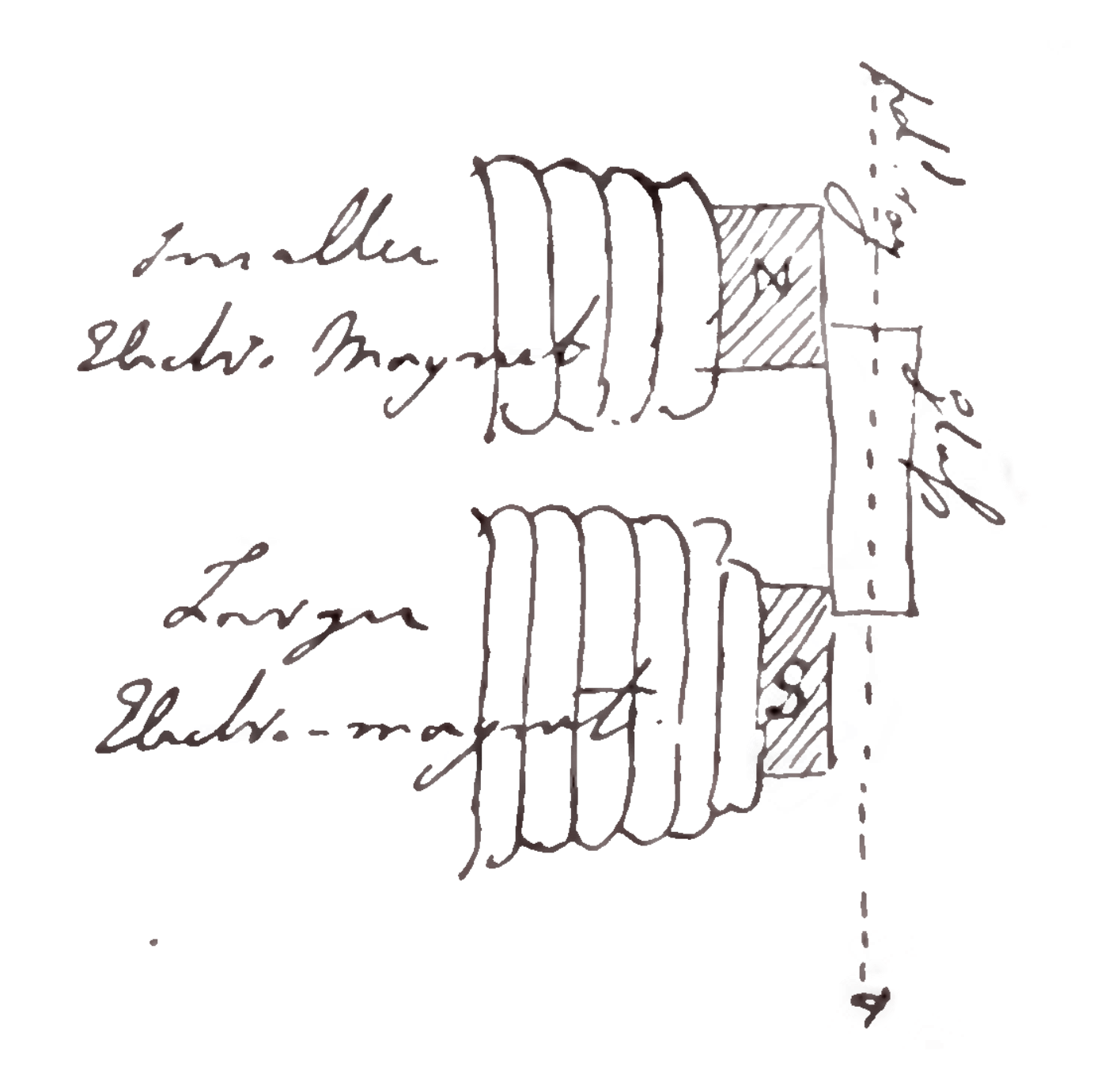}
  \caption{Figure from Michael Faraday's diary~\cite{FaradaysDiary4}, September 13, 1845. A piece of transparent glass (``silico borate of lead'') rotated the plane of polarization of the light going through the glass in the same direction as the applied magnetic field. Faraday did not understand at the time the mechanism behind this phenomenon, but he knew that the change in polarization was a probe codifying important information on the ray of light, the nature of the magnetic field, and the material itself. After many other tests, he famously wrote in entry 7718, September 30, 1845: ``Still, I have at last succeeded in illuminating a magnetic curve or line of force and in magnetizing a ray of light.'' This experiment was the first indication of a deeper relationship between optics and electromagnetism, two decades before Maxwell's prediction~\cite{Maxwell1865} of light being an EMW in 1865.}
  \label{Fig_Faraday}
\end{figure}


\begin{acknowledgments}
The authors wish to thank Yuri Bonder, Jens Boos, Fabrizio Canfora, Oscar Castillo-Felisola, Norman Cruz, Nicolás González, Julio Oliva, Miguel Pino, Francisca Ramírez, Patricio Salgado, Sebastián Salgado, Jorge Zanelli, and Alfonso Zerwekh for enlightening discussions and comments.
J.~B. acknowledges financial support from Comisión Nacional de Investigación Científica y Tecnológica (CONICYT) grant 21160784 and also thanks the Institute of Mathematics of the Czech Academy of Sciences, where part of this work was carried out.
F.~C.-T. was supported by Comisión Nacional de Investigación Científica y Tecnológica (CONICYT) grant 72160340. He wishes to thank Dieter Lüst for his kind hospitality at the Arnold Sommerfeld Center for Theoretical Physics in Munich.
The work of C.~C. is supported by Proyecto POSTDOC\_DICYT, Código 041931CM\_POSTDOC, Universidad de Santiago de Chile.
F.~I. is funded by Fondo Nacional de Desarrollo Científico y Tecnológico (FONDECYT) grant 1180681.
P.~M. is supported by Comisión Nacional de Investigación Científica y Tecnológica (CONICYT) grant 21161574.
O.~V. is supported by Project VRIIP0062-19, Universidad Arturo Prat, Chile.
\end{acknowledgments}


\appendix

\section{Detailed derivation of the wave equation}
\label{Apendice}

Replacing Eqs.~(\ref{Eq_Perturb_R1}), (\ref{Eq_Perturb_U1}) and~(\ref{Eq_D-Da}) in Eq.~(\ref{Eq_High_Frequency}), it is straightforward to show that
\begin{equation}
  \frac{1}{2} \epsilon_{abcn} R_{\left( 1 \right)}^{ab} \wedge e^{c} =
  \left( W_{mn} - \frac{1}{2} \eta_{mn} W^{p}{}_{p} \right) \hodge e^{m},
\end{equation}
where
\begin{equation}
  W^{m}{}_{n} = \left(
    \mathrm{I}_{n} \mathcal{D}_{a} -
    \mathrm{I}_{a} \mathcal{D}_{n} +
    T^{p}{}_{na} \mathrm{I}_{p}
  \right)
  \left[ U_{\left( 1 \right)}^{am} + V^{am} \right].
\end{equation}
Using the commutator from Eq.~(\ref{Eq_Super[Ia,Db]}),
\begin{equation}
  T^{p}{}_{na} \wedge \mathrm{I}_{p} =
  \mathrm{I}_{a} \mathcal{D}_{n} - \mathcal{D}_{n} \mathrm{I}_{a},
\end{equation}
$W^{m}{}_{n}$ is given by
\begin{equation}
  W^{m}{}_{n} = \left(
    \mathrm{I}_{n} \mathcal{D}_{a} -
    \mathcal{D}_{n} \mathrm{I}_{a}
  \right)
  \left[ U_{\left( 1 \right)}^{am} + V^{am} \right].
\end{equation}

It is possible to prove that the $U_{ab}^{\left( 1 \right)}$ term in Eq.~(\ref{Eq_Perturb_U1}) can be rewritten as
\begin{equation}
  U_{ab}^{\left( 1 \right)} =
  -\frac{1}{2} \left( \mathcal{D}_{a} H_{b} - \mathcal{D}_{b} H_{a} \right),
\end{equation}
and therefore
\begin{align}
  W_{mn} & = \frac{1}{2} \Big[
    -\mathrm{I}_{n} \mathcal{D}_{a} \mathcal{D}^{a} H_{m} +
    \mathrm{I}_{n} \mathcal{D}_{a} \mathcal{D}_{m} H^{a} 
   \nonumber \\ &\quad  +
    \mathcal{D}_{n} \mathrm{I}_{a} \left(
      \mathcal{D}^{a} H_{m} - \mathcal{D}_{m} H^{a}
    \right)
  \Big]
  \nonumber \\ &\quad +
  \left(
    \mathrm{I}_{n} \mathcal{D}_{a} - \mathcal{D}_{n} \mathrm{I}_{a}
  \right) V^{a}{}_{m}.
\end{align}

The ``Lorenz gauge fixing'' in Eq.~(\ref{Eq_Lorenz_Gauge}) can be rewritten as
\begin{equation}
  \mathrm{I}_{m} \mathcal{D}_{a} H^{a} +
  \mathrm{I}_{a} \left( \mathcal{D}^{a} H_{m} - \mathcal{D}_{m} H^{a} \right) = 0,
\end{equation}
and therefore
\begin{align}
  W_{mn} & = \frac{1}{2} \Big[
    -\mathrm{I}_{n} \mathcal{D}_{a} \mathcal{D}^{a} H_{m} +
    \mathrm{I}_{n} \left[
      \mathcal{D}_{a}, \mathcal{D}_{m}
    \right] H^{a} 
  \nonumber \\ \label{eqa8} &\quad  +
    \left(
      \mathrm{I}_{n} \mathcal{D}_{m} -
      \mathcal{D}_{n} \mathrm{I}_{m}
    \right)
    \mathcal{D}_{a} H^{a}
  \Big] + \left(
    \mathrm{I}_{n} \mathcal{D}_{a} -
    \mathcal{D}_{n} \mathrm{I}_{a}
  \right) V^{a}{}_{m}.
\end{align}

However, using Eq.~(\ref{Eq_Def_Da}) it is straightforward to show that
\begin{equation}
  \mathrm{I}_{n} \mathcal{D}_{m} -
  \mathcal{D}_{n} \mathrm{I}_{m} =
  \mathrm{I}_{mn} \mathrm{D} - \mathrm{DI}_{mn},
\end{equation}
and since from the ``Lorenz gauge fixing''
$\mathcal{D}_{a} H^{a} = \frac{1}{2} \mathrm{d} H$,
in this gauge we have that
\begin{equation}
\left(  \mathrm{I}_{n}\mathcal{D}_{m}-\mathcal{D}_{n}\mathrm{I}_{m}\right)\mathcal{D}_{a}H^{a}  =\frac{1}{2}\left(  \mathrm{I}_{mn}\mathrm{D}-\mathrm{DI}_{mn}\right)  \mathrm{d}H=0.
\end{equation}
Inserting this result into Eq.~\eqref{eqa8}, we get the inhomogeneous wave equation
\begin{align}
  W_{mn} & = \frac{1}{2} \left(
    -\mathrm{I}_{n} \mathcal{D}_{a} \mathcal{D}^{a} H_{m} +
    \mathrm{I}_{n} \left[
      \mathcal{D}_{a}, \mathcal{D}_{m}
    \right] H^{a}
  \right) 
  \nonumber \\ &\quad + \left(
    \mathrm{I}_{an} \mathrm{D} - \mathrm{DI}_{an}
  \right) V^{a}{}_{m}
  \nonumber \\ & =
  \frac{1}{2} \mathrm{I}_{n} \left(
    -\mathcal{D}_{a} \mathcal{D}^{a} H_{m} + \left[
      \mathcal{D}_{a}, \mathcal{D}_{m}
    \right] H^{a} +
    2\mathrm{I}_{a} \mathrm{D} V^{a}{}_{m}
  \right).
\end{align}


\bibliography{GB2019_arXiv_v3}


\end{document}